\documentclass[12pt]{article}

\usepackage{amssymb,amsthm,amscd, amsbsy, array}
\usepackage[all]{xy}
\usepackage{revsymb}

\usepackage[cyr]{aeguill}     
\usepackage[pdftex]{graphicx} 
\DeclareGraphicsExtensions{.jpg,.mps,.pdf,.png} 

\let\ssection=\section
\renewcommand{\section}{\setcounter{equation}{0}\ssection}

\newcommand{\la}{{\langle}}
\newcommand{\ra}{{\rangle}}

\newcommand{\bara}{{\bar{a}}}
\newcommand{\barA}{{\bar{A}}}

\newcommand{\Ad}{\mathrm{Ad}}

\newcommand{\barb}{{\bar{b}}}
\newcommand{\barB}{{\bar{B}}}

\newcommand{\barc}{{\bar{c}}}

\newcommand{\Coad}{{\mathrm{Coad}}}

\newcommand{\bard}{{\bar{d}}}

\newcommand{\dirsupp}{u}

\newcommand{\be}{{\mathbf{e}}}

\newcommand{\se}{\mathrm{e}}

\newcommand{\he}{\hat{e}}
\newcommand{\rE}{\mathrm{E}}

\newcommand{\GL}{{\mathrm{GL}}}

\newcommand{\rg}{\mathrm{g}}

\newcommand{\fg}{\mathfrak{g}}

\newcommand{\bdirsupp}{{\mathbf{\dirsupp}}}

\newcommand{\cM}{\mathcal{M}}

\newcommand{\sfn}{\mathsf{n}}

\newcommand{\barn}{{\bar{n}}}

\newcommand{\hnabla}{\widehat{\nabla}}

\newcommand{\SO}{\mathrm{SO}}
\newcommand{\so}{\mathrm{o}}
\newcommand{\cO}{{\mathcal{O}}}

\newcommand{\bomega}{\boldsymbol{\omega}}
\newcommand{\bhomega}{\widehat{\boldsymbol{\omega}}}
\newcommand{\homega}{\widehat{\omega}}
\newcommand{\hOmega}{\widehat{\Omega}}

\newcommand{\bP}{{\bf P}}

\newcommand{\hP}{{\widehat{P}}}

\newcommand{\bq}{{\mathbf{q}}}

\newcommand{\hQ}{{\widehat{Q}}}

\newcommand{\bbR}{\mathbb{R}}
\newcommand{\hR}{{\widehat{R}}}

\newcommand{\bS}{{\mathbf{S}}}

\newcommand{\SE}{\mathrm{SE}}

\newcommand{\sign}{\mathrm{sign}}

\newcommand{\TMM}{TM\!\setminus\!M}
\newcommand{\Tr}{\mathrm{Tr}}

\newcommand{\hu}{{\hat{u}}}
\newcommand{\bhu}{\mathbf{\hat{u}}}
\newcommand{\bu}{{\mathbf{u}}}
\newcommand{\hdirsupp}{{\hat{\dirsupp}}}

\newcommand{\Vect}{\mathrm{Vect}}
\newcommand{\Vol}{\mathrm{Vol}}
\newcommand{\vol}{\mathrm{vol}}
\newcommand{\bvol}{\mathbf{vol}}
\newcommand{\dv}{\dot{v}}

\newcommand{\bx}{{\mathbf{x}}}

\newcommand{\dx}{\dot{x}}

\newcommand{\by}{{\mathbf{y}}}

\newcommand{\dy}{\dot{y}}

\newcommand{\half}{\frac{1}{2}}

\newtheorem{thm}{Theorem}[section]
\newtheorem{lem}[thm]{Lemma}
\newtheorem{cor}[thm]{Corollary}
\newtheorem{pro}[thm]{Proposition}

\newtheorem{rmk}[thm]{Remark}
\newtheorem{defi}[thm]{Definition}
\newtheorem{ax}[thm]{Axiom}

\newcommand{\KEYWORDS}[1]{\textbf{Keywords}: #1}
\newcommand{\AMSCLASSIFICATION}[1]{\textbf{Mathematics Subject Classification 2000}: #1}
\newcommand{\PREPRINT}[1]{\textbf{Preprint}: #1}

\title{Finsler Spinoptics}
\author{
C. DUVAL\footnote{mailto: duval@cpt.univ-mrs.fr}\\\\
Centre de Physique Th\'eorique, CNRS, 
Luminy, Case 907\\ 
F-13288 Marseille Cedex 9 (France)\footnote{ 
UMR 6207 du CNRS associ\'ee aux 
Universit\'es d'Aix-Marseille I et II et Universit\'e du Sud Toulon-Var; Laboratoire 
affili\'e \`a la FRUMAM-FR2291
}
}
\date{March 14 2008}

\begin{document}

\maketitle
\thispagestyle{empty}

\abstract{
The objective of this article is to build up a general theory of geometrical optics for spinning light rays in an inhomogeneous and anisotropic medium modeled on a Finsler manifold. The prerequisites of local Finsler geometry are reviewed together with the main properties of the Cartan connection used in this work. Then, the principles of Finslerian spinoptics are formulated on the grounds of previous work on Riemannian spinoptics, and relying on the generic coadjoint orbits of the Euclidean group. A new presymplectic structure on the indicatrix-bundle is introduced, which gives rise to a foliation that significantly departs from that generated by the geodesic spray, and leads to a specific anomalous velocity, due to the coupling of spin and the Cartan curvature, and related to the optical Hall effect.
}

\vskip 6cm
\KEYWORDS{Presymplectic manifolds, Geometrical spinoptics, Finsler structures}

\medskip

\AMSCLASSIFICATION{78A05, 70G45, 58B20}

\medskip

\PREPRINT{CPT-P29-2007}

\newpage

\tableofcontents

\baselineskip=22pt


\section{Introduction}

Geometry and optics have maintained a lasting relationship since Euclid's Optics where light rays were first interpreted as oriented straight lines in space (or, put in modern terms, as oriented, non parametrized, geodesics of Euclidean space). One can, withal, trace back the origin of the calculus of variations to Fermat's Principle of least optical path. This principle has served as the basis of \textit{geometrical optics} in inhomogeneous, isotropic, media and proved a fundamental mathematical tool in the design of optical (and electronic) devices such as mirrors, lenses, etc., and in the understanding of caustics, and optical aberrations. 

Although Maxwell's theory of wave optics has unquestionably clearly superseded geometrical optics as a bona fide theory of light, the seminal work of Fermat has opened the way to wide branches of mathematics, physics, and mechanics, namely, to the calculus of variations in the large, modern classical (and quantum) field theory, Lagrangian and Hamiltonian or presymplectic mechanics.

It should be stressed that Fermat's Principle has, in essence, a close relationship to modern \textit{Finsler geometry}, as it rests on a specific ``Lagrangian''
\begin{equation}
F(x,y)=\sfn(x)\sqrt{\delta_{ij}\,y^iy^j}
\label{FermatLagangian}
\end{equation}
where $y=(y^i)$ stands for the ``velocity'' of light, and $\sfn(x)>0$ for the value of the (smooth) refractive index of the medium, at the ``location'' $x$ in Euclidean space. (Note that Einstein's summation convention is tacitly understood throughout this article.)
As a matter of fact, the function (\ref{FermatLagangian}) is a Finsler metric, namely a positive function, homogeneous of degree one in the velocity, smooth wherever $y\neq0$, and such that the Hessian $\rg_{ij}(x,y)=\left(\half{}F^2\right)_{y^iy^j}$ is positive-definite (see Section~\ref{SubSectionFinslerMetrics}). Although, this is a very special case of Finsler metric --- it actually defines a conformal\-ly flat Riemannian metric tensor, viz., $\rg_{ij}(x)=\sfn^2(x)\,\delta_{ij}$ ---, this fact is worth noting for further generalization. The geodesics of the Fermat-Finsler metric (\ref{FermatLagangian}) are a fairly good mathematical model for light rays in refractive, inhomogeneous, and non dispersive media --- provided polarization of light is ignored!

It has quite recently been envisaged to consider a general Finsler metric~$F(x,y)$ to describe anisotropy of optical media, as the Finsler metric tensor, $\rg_{ij}(x,y)$, depends, in general non trivially, on the direction of the velocity, $y$, or  ``\'el\'ement de support'' in the sense of Cartan \cite{Car}. This enables one to account for the fact that~\cite{AIM,Ing}
\begin{quote}
{\sl
In an anisotropic medium, the speed of light depends on its direction, and the unit surface is no longer a sphere.}
(Finsler, 1969)
\end{quote}
The Fermat Principle has also been reformulated in the pre\-symplectic framework in~\cite{CN1}, and generalized in \cite{CN2} to the context of anisotropic media. By the way, the regularity condi\-tion imposed, in the latter reference,  amounts to demanding a Finsler structure.

One thus takes for granted that oriented Finsler geodesics may describe light rays in an\-isotropic media.

\goodbreak

A Finslerian version of the Fermat Principle now states that the second-order differential equations governing the propagation of light stem from the geodesic spray of a (three-dimensional) Finsler space, $(M,F)$, given by the Reeb vector field of the contact $1$-form
\begin{equation}
\varpi=\omega_3,
\label{FinslerFermat2form}
\end{equation}
where $\omega_3=F_{y^i}dx^i$ is, here, the restriction to the indicatrix-bundle, $SM=F^{-1}(1)$, of the Hilbert $1$-form. See Section \ref{SectionGeometricalOpticsFinsler}, and also \cite{Fou}.

Now, geometrical optics is, from a different standpoint, widely accepted as a semi-classical limit \cite{BW} of wave optics with ``small parameter'' the reduced wavelength $\lambdabar$ (or~$\lambdabar/L$, where~$L$ is some characteristic length of the optical medium, see, e.g., \cite{BFK}). It has, however, recently been established on experimental grounds that trajectories of light beams in inhomogeneous optical media depart from those predicted by geometrical optics. See, e.g. \cite{BB04,BB05,Bli}, and \cite{OMN,OMNbis}, for several approaches to photonic dynamics in terms of a semi-classical limit of the Maxwell equations in inhomogeneous, and isotropic media, highlighting the Berry connection~\cite{Ber}. See also \cite{GBM}. The so-called \textit{optical Hall effect} for polarized light rays, featuring a very small transverse shift, ortho\-gonal to the gradient of the refractive index, has, hence, received a firm theoretical explanation.

From quite a different perspective, a theory of spinning light in arbitrary three-dimensional Riemannian  manifold has been put forward as an extension of the Fermat Principle to classical, circularly polarized photons, in inhomogeneous, (es\-sentially) isotropic, media. This theory of \textit{spinoptics}, presented in \cite{DHH}, and \cite{DHH2}, relies fundamentally on the Euclidean group, $\rE(3)$, viewed at the same time as the group of isometries of Euclidean space, $E^3$, and as the group of symmetries of classical states of free photons represented by Euclidean coadjoint orbits. Straight\-forward adaptation of the general relativistic prescription of minimal coupling~\cite{Kun,Sou2,Ste} readily yielded a set of differential equations governing the trajectories of spinning light in inhomogeneous, and isotropic media described by a Riemannian structure. Also did this formalism for spinoptics help to put the optical Hall effect in proper geometrical perspective, in agreement with \cite{OMN}.

\goodbreak

The main purport of the present article is, as might be expected, to try and provide a fairly natural extension of plain geometrical optics --- in non-dispersive, anisotropic, media described in terms of Finsler geodesics --- to spinoptics, i.e., to the case of \textit{circularly polarized light rays} carrying color and helicity in such general optical media. In doing so, one must unavoidably choose a linear Finsler connection from the start (see (\ref{FinslerFermatIntro})), the crux of the matter being that there is, apart from the special Riemannian case, no canonical Finsler connection at hand. The challenge may, in fact, be accepted once we take seriously the Euclidean symmetry as a guiding principle, a procedure that can be implemented by considering the dipole approximation to ordinary geometrical optics, namely the spinning coadjoint orbits of the Euclidean group. This is the subject of Section \ref{SectionFinslerSpinoptics} which contains the main results of this article, where the Finsler-Cartan connection prevailed definitely over other Finsler connections, as regard to the original, fundamental, Euclidean symmetry of the free model. Let us, however, mention that all resulting expressions, for the foliations we end up with, ultimately depend upon the Finsler metric tensor, the Cartan tensor, and the Chern curvature tensors only.

The hereunder proposed principle of \textit{Finsler spinoptics} (see Axiom \ref{AxiomSpinoptics}, which can be understood the prescription of minimal coupling to a Finsler-Cartan connection) thus amounts to consider, instead of (\ref{FinslerFermat2form}), the following $1$-form
\begin{equation}
\varpi = \omega_3+\lambdabar\,\homega_{12},
\label{FinslerFermatIntro}
\end{equation}
where $\lambdabar$ is the (signed) \textit{wavelength}, the $\homega_{ab}$, with $a,b=1,2,3$, representing the compo\-nents of the Cartan connection associated with  a three-dimensional Finsler manifold $(M,F)$. The $1$-form (\ref{FinslerFermatIntro}) might be considered as providing a \textit{deformation} of the Hilbert $1$-form driven by the wavelength parameter, $\lambdabar$. See Remark \ref{RemarkBerry} below.

\goodbreak

The characteristic foliation of the novel $2$-form $\sigma=d\varpi$ is explicitly calculated, and leads to a drastic deviation from the Finsler geodesic spray, dictated by spin-curvature coupling terms which play, in this formalism, quite a significant r\^ole, as expressed by Theorem \ref{MainTheorem}. Of course, the equations of spinoptics in a Riemannian manifold \cite{DHH} are recovered, as special case of those corresponding to a Finsler-Cartan structure. We assert that this foliation can be considered a natural extension of the Finsler geodesics spray to the case of spinoptics in Finsler-Cartan spaces.

The paper is organized as follows.

Section \ref{SectionFinslerStructures} provides a survey of local Finsler geometry. We found it necessary to offer a somewhat technical and detailed introduction of the objects pertaining to Finsler geometry, in particular to the various connections used throughout this article, to make the reading easier to non experts. Emphasis is put on the Chern and Cartan connections, as these turn out to be of central importance in this study. This section relies essentially on the authoritative Reference~\cite{BCS}.

In Section \ref{SectionGeometricalOpticsFinsler}, we review the principles of geometrical optics, extending Fermat's optics to the area of Finsler structures characterizing anisotropic optical media. Then, special attention is paid to the Hilbert $1$-form in the derivation of the Finsler geodesic spray. The connection of the latter to the Fermat differential equations associated with conformally related Finsler structures is furthermore analyzed.

Section \ref{SectionFinslerSpinoptics} constitutes the major part of the article. It presents, in some details, the basic structures arising in the classification of the $\SE(3)$-homogeneous symplectic spaces, which are interpreted as the seeds of spinoptics, namely the Euclidean coadjoint orbits labeled by color, and spin, according to the classic \cite{Sou}. The core of our study consists in the choice of a special Finsler connection, namely the Finsler-Cartan connection, to perform minimal coupling of spinning light particles to a Finsler metric. This is done and explained in this section, in which the derivation of the characteristic foliation of our distinguished presymplectic $2$-form $d\varpi$, see (\ref{FinslerFermatIntro}), is spelled out in detail. This completes the introduction of Finslerian spinoptics.
 
\goodbreak
 
Conclusions are drawn in Section \ref{SectionConclusion}, and perspectives for future work connected to the present study are finally outlined.

\medskip
\textbf{Acknowledgments}: It is a great pleasure to thank J.-C.~Alvares Paiva, S.~Tabachnikov, and P.~Verovic, for useful correspondence and enlightening discussions. Thanks are also due to P.~Horv\'athy for valuable advice.


\section{Finsler structures: a compendium}\label{SectionFinslerStructures}

\subsection{Finsler metrics}\label{SubSectionFinslerMetrics}

\subsubsection{An overview}

A Finsler structure is a pair $(M,F)$ where $M$ is a smooth, $n$-dimensional, manifold and $F:TM\to\bbR^+$ a given function whose restriction to the slit tangent bundle $TM\!\setminus\!M=\{(x,y)\in{}TM\,\vert\,y\in{}T_xM\!\setminus\!\{0\}\}$ is smooth, and (fiberwise) positively homogeneous of degree one, i.e., $F(x,\lambda y)=\lambda F(x,y)$, for all $\lambda>0$; one furthermore demands that the $n\times{}n$ Hessian matrix with entries 
\begin{equation}
\rg_{ij}(x,y)=\left(\half{}F^2\right)_{y^i y^j}
\label{gij}
\end{equation}
be positive-definite
for all $(x,y)\in\TMM$. 
The quantities $\rg_{ij}$ defined in (\ref{gij}) are (fiberwise) homogeneous of degree zero, and 
\begin{equation}
\rg=\rg_{ij}(x,y)dx^i\otimes{}dx^j
\label{g}
\end{equation}
defines a \textit{sphere's worth of Riemannian metrics} \cite{BR} on each $T_xM$ para\-metrized by the direction of $y$. We will put $(\rg^{ij})=(\rg_{ij})^{-1}$. If $\pi:\TMM\to{}M$ stands for the canonical surjection, the metric (or fundamental) ``tensor'' (\ref{g}) is, actually, a section of the bundle $\pi^*(T^*M)\otimes\pi^*(T^*M)\to\TMM$. 

The distinguished ``vector field'' (the direction of the \textit{the supporting element}) 
\begin{equation}
\dirsupp=\dirsupp^i\frac{\partial}{\partial{x^i}},
\qquad
\mathrm{where}
\qquad
\dirsupp^i(x,y)=\frac{y^i}{F(x,y)},
\label{ell}
\end{equation}
is, indeed, a section of $\pi^*(TM)\to\TMM$ such that 
$\rg(\dirsupp,\dirsupp)=\rg_{ij}\dirsupp^i\dirsupp^j=1$.

\goodbreak

There is, at last, another tensor specific to Finsler geometry, namely the \textit{Cartan tensor}, 
$C=C_{ijk}(x,y)dx^i\otimes{}dx^j\otimes{}dx^k$,
where
$C_{ijk}(x,y)
=\left(\frac{1}{4}F^2\right)_{y^i y^j y^k}$.
As in~\cite{BCS}, we will also use \textit{ad lib.} the quantities 
\begin{equation}
A_{ijk}=F\,C_{ijk},
\label{A}
\end{equation} 
which are totally symmetric, 
$A_{ijk}=A_{(ijk)}$,
and enjoy the following property, viz.,
\begin{equation}
A_{ijk}\,\dirsupp^k=0.
\label{Al=0}
\end{equation}

\goodbreak

There is a wealth of Finsler structures, apart from the well-known special case of Riemannian structures $(M,\rg)$ for which $F(x,y)=\sqrt{\rg_{ij}(x)y^iy^j}$. 
See, e.g., \cite{BCS,BR,Shen} for a survey, and for a list of examples of Finsler structures.
We will review below, see (\ref{Fe})--(\ref{Fpm}), examples of Finsler structure associated with optical birefringence~\cite{AIM}.

\subsubsection{Introducing special orthonormal frames}

Having chosen a coordinate system ($x^i$) of $M$, we denote --- with a slight abuse of notation --- by $\partial/\partial{x^i}$ (resp. $dx^i$) the so-called \textit{transplanted sections} of $\pi^*(TM)$ (resp. $\pi^*(T^*M)$), 
Accordingly we will denote by $\partial/\partial{y^i}$ (resp.~$dy^i$) the standard, \textit{vertical}, sections of~$T(\TMM)$ (resp. $T^*(\TMM)$).



Introduce now special $\rg$-orthonormal frames $(e_1,\ldots,e_n)$ for $\pi^*(TM)$, such that
\begin{equation}
\rg(e_a,e_b)=\delta_{ab},
\label{gab}
\end{equation}
for all $a,b=1,\ldots,n$, with, as preferred element, the distinguished section
\begin{equation}
e_n=\dirsupp.
\label{en}
\end{equation}
Recall that each $e_a$ lies in the fiber $\pi^*(TM)_{(x,y)}$ above $(x,y)\in\TMM$. 
The local decomposition of these vectors is given by
\begin{equation}
e_a=e_a^i\frac{\partial}{\partial{x^i}},
\label{eaLoc}
\end{equation}
for all $a=1,\ldots,n$, where the matrix $(e_a^i)$, defined at $(x,y)\in\TMM$, is nonsingular. 
We thus have
\begin{equation}
e_n^i=\dirsupp^i.
\label{eni}
\end{equation}

The dual frames, for $\pi^*(T^*M)$, which we denote by $(\omega^1,\ldots,\omega^n)$, are such that $\omega^a(e_b)=\delta^{a}_b$, for all $a,b=1,\ldots,n$. Accordingly, we have the local decompo\-sition
\begin{equation}
\omega^a=\omega^a_i dx^i,
\label{omegaaLoc}
\end{equation}
for all $a=1,\ldots,n$, where $(\omega^a_i)=(e_a^i)^{-1}$.
The $1$-form dual to $e_n$ is 
the \textsl{Hilbert form}
\begin{equation}
\omega_H=\omega^n
\label{Hilbert}
\end{equation}
which, in view of (\ref{ell}), reads
$\omega_H=\dirsupp_i\,dx^i$,
with
$\dirsupp_i=\rg_{ij}\dirsupp^j=F_{y^i}$. 


The following proposition introduces the principal bundle of orthonormal frames above the slit tangent bundle of a Finsler manifold.
\begin{pro}
The manifold, $\SO_{n-1}(\TMM)$, of special $\rg$-orthonormal frames for $\pi^*(TM)$ is endowed with a structure of $\SO(n-1)$-principal bundle over $\TMM$. If $\GL_+(M)$ stands for the bundle of positively oriented linear frames of $M$, it is defined by $\SO_{n-1}(\TMM)=\Psi^{-1}(0)$ where $\Psi:\pi^*(\GL_+(M))\to\bbR^{\half{}n(n+1)}\times\bbR^{n-1}$ is given by
$\Psi(x,y,(e_a))=((\rg_{(x,y)}(e_a,e_b)-\delta_{ab}),y-F(x,y)e_n)$,
where $a,b=1,\ldots,n$.
\end{pro}
The proof is straightforward as is that of the following statement.
\begin{cor}\label{CorrSOSM}
Let $SM=F^{-1}(1)$ denote the indicatrix-bundle of a Finsler manifold~$(M,F)$, and $\iota:SM\hookrightarrow\TMM$ its embedding into the tangent bundle of $M$. The pull-back $\SO_{n-1}(SM)=\iota^*(\SO_{n-1}(\TMM))$ is a principal $\SO(n-1)$-bundle over~$SM$.
\end{cor}

\subsubsection{The non-linear connection}

Let us recall that the Finsler metric, $F$, induces in a canonical fashion a splitting of the tangent bundle of the slit tangent bundle $\pi:\TMM\to{}M$ of~$M$ as follows:
\begin{equation}
T(\TMM)=V(\TMM)\oplus{}H(\TMM),
\label{Splitting}
\end{equation} 
where the vertical tangent bundle is $V(\TMM)=\ker\pi_*$. The fibers, $V_{(x,y)}$, of that subbundle are spanned by the \textit{vertical} local basis vectors
$
\partial/\partial{}y^i,
$
with $i=1,\ldots,n$. Those, $H_{(x,y)}$, of the \textit{horizontal} subbundle $H(\TMM)$ are spanned by the \textit{horizontal} local basis vectors
\begin{equation}
\frac{\delta}{\delta{}x^i}=\frac{\partial}{\partial{}x^i}-N^{j}_{\;i}\frac{\partial}{\partial{}y^j},
\label{deltaxi}
\end{equation} 
with $i=1,\ldots,n$, where the $N_i^{\;j}$ are the coefficients of a \textit{non-linear connection} canonically defined by 
\begin{equation}
N^i_{\;j}=\half\frac{\partial{}G^i}{\partial{}y^j}
\label{Nij}
\end{equation} 
in terms of the \textit{spray} coefficients
$
G^j=\half\rg^{jk}\left((F^2)_{y^k{}x^l}\,y^l-(F^2)_{x^k}\right)
$.


The horizontal vectors, $\delta/\delta{}x^i$, see (\ref{deltaxi}), and the vertical vectors, $\partial/\partial{}y^i$, 
form a local natural basis for $T_{(x,y)}(\TMM)$ whose dual basis is given by $dx^i$ and $\delta{}y^i$, where
\begin{equation}
\delta{y^i}=dy^i+N^i_{\;j}dx^j.
\label{deltay}
\end{equation}

We now introduce, following (\ref{eaLoc}), (\ref{omegaaLoc})
the $\rg$-orthonormal frames that will be needed in the sequel.
\begin{defi}
We will call $\mathrm{hv}$-frame any frame for $T(\TMM)$, compatible with the splitting (\ref{Splitting}), namely \cite{BCS}
\begin{equation}
\he_a=e^i_a\frac{\delta}{\delta{}x^i},
\qquad
\qquad
\he_\bara=e^i_a\,F\frac{\partial}{\partial{}y^i},
\label{orthonormalBases}
\end{equation}
where $\bara=a+n$, with $a=1,\ldots,n$. The associated dual basis reads then
\begin{equation}
\omega^a=\omega_i^a\,dx^i,
\qquad
\qquad
\omega^\bara=\omega_i^a\,\frac{\delta{}y^i}{F},
\label{dualOrthonormalBases}
\end{equation}
with $a=1,\ldots,n$; see (\ref{omegaaLoc}) and (\ref{deltay}).
\end{defi}

Owing to the properties of the non-linear connection (\ref{Nij}), the metric, $F$, is horizontally constant, 
\begin{equation}
\frac{\delta{F}}{\delta{x^i}}=0, 
\label{deltaFdeltax=0}
\end{equation}
for all $i=1,\ldots,n$. 
This entails that
the vertical $1$-form $\omega^\barn$ is exact,
\begin{equation}
\omega^\barn=d\log{F}.
\label{omegabn}
\end{equation}

\subsection{Finsler connections}

Unlike the Riemannian case, there is no canonical linear Finsler connection on the bundle $\pi^*(TM)$, which is required as soon as one needs to differentiate tensor fields, i.e., sections of the bundles $\pi^*(TM)^{\otimes{}p}\otimes\pi^*(T^*M)^{\otimes{}q}$. 

\subsubsection{The Chern connection}

A celebrated example, though, is the \textit{Chern connection} $\omega^{\;i}_j=\Gamma^i_{\;jk}(x,y)dx^k$ which is uniquely defined by the following requirements \cite{BCS}:
(i) it is symmetric: $\Gamma^i_{\;jk}=\Gamma^i_{\;kj}$, and
(ii) it \textit{almost} transports the metric tensor: $dg_{ij}-\omega^{\;k}_ig_{jk}-\omega^{\;k}_jg_{ik}=2\,C_{ijk}\delta{y^k}$,
where the $\delta{}y^k$ are as in (\ref{deltay}).

\goodbreak


The Chern connection coefficients turn out to yield
the alternative expression of the non-linear connection, namely,
\begin{equation}
N^i_{\;j}(x,y)=\Gamma^i_{\;jk}y^k.
\label{NijChern}
\end{equation}

The co\-variant derivative $\nabla:\Gamma(\pi^*(TM))\to\Gamma(T(\TMM)\otimes\pi^*(TM))$ associated with the Chern connection is related to the 
$\omega^{\;i}_j$ via $\nabla_X\partial/\partial{}x^j=\omega^{\;i}_j(X)\partial/\partial{}x^i$, for all $X\in\Vect(\TMM)$.
In terms of the special $\rg$-orthonormal frames, we can write $\nabla_X{}e_b=\omega^{\;a}_b(X)e_b$, for all $X\in\Vect(\TMM)$ where
$\omega^{\;a}_b=\omega^a_i(de^i_b+\omega^{\;i}_j e^j_b)$
denote the frame components of the Chern connection.
The following theorem summarizes the defining properties of the Chern con\-nection. 
\begin{thm}\cite{BCS}
There exists a unique linear connection, $(\omega^a_b)$, on $\pi^*(TM)$, named the Chern connection, which is torsionfree and almost $\rg$-compatible, namely
\begin{equation}
\Omega^a=d\omega^a-\omega^b\wedge\omega^{\;a}_b=0
\label{Omega=0}
\end{equation}
and
\begin{equation}
\omega_{ab}+\omega_{ba}=-2A_{abc}\,\omega^{\barc},
\label{omegaabSym}
\end{equation}
where $a,b,c=1,\ldots,n$.
The corresponding connection coefficients are of the form
\begin{equation}
\Gamma^i_{\;jk}
=
\half\rg^{il}\left(\frac{\delta\rg_{kl}}{\delta{x^j}}+\frac{\delta\rg_{jl}}{\delta{x^k}}-\frac{\delta\rg_{jk}}{\delta{x^l}}\right),
\label{GammaChern}
\end{equation}
for all $i,j,k=1,\ldots,n$.
\end{thm}


The curvature, $(\Omega^a_b)$, of the linear connection, $(\omega^{\;a}_b)$, is defined by the structure equations
$\Omega^{\;a}_b=d\omega^{\;a}_b-\omega^{\;c}_b\wedge\omega^{\;a}_c$.
It retains the form
\begin{equation}
\Omega^{\;a}_b=\half{}R^{\;a}_{b\;cd}\,\omega^b\wedge\omega^d+P^{\;a}_{b\;cd}\,\omega^b\wedge\omega^\bard,
\label{Omega}
\end{equation}
where 
$R^{\;i}_{j\;kl}=e^i_a\,\omega^b_j\,\omega^c_k\,\omega^d_l\,R^{\;a}_{b\;cd}$,
reads
\begin{equation}
R^{\;i}_{j\;kl}
=
\frac{\delta\Gamma^i_{\;jl}}{\delta{x^k}}
-
\frac{\delta\Gamma^i_{\;jk}}{\delta{x^l}}
+\Gamma^i_{\;mk}\Gamma^m_{\;jl}
-
\Gamma^i_{\;ml}\Gamma^m_{\;jk},
\label{R}
\end{equation}
and, accordingly,
\begin{equation}
P^{\;i}_{j\;kl}
=
-F\frac{\partial\Gamma^i_{jk}}{\partial{y^l}}.
\label{P}
\end{equation}

The $\mathrm{hv}$-curvature, $P$, enjoys the fundamental property
\begin{equation}
P^{\;i}_{j\;kl}\,\dirsupp^l=0.
\label{Pl=0}
\end{equation}

Using Cartan's notation \cite{Car} (see also \cite{Run,BCS}), we write the covariant derivative of, e.g., a section $X$ of $\pi^*(TM)$ as
$(\nabla{X})^i=dX^i+\omega^{\;i}_jX^j=X^i_{\vert{}j}\,dx^j+X^i_{\Vert{}j}F^{-1}\delta{}y^j$,
where $i,j=1,\ldots,n$. In particular, it can be easily deduced from (\ref{NijChern}) that the covariant derivative of the unit vector~$\dirsupp$ is given by
\begin{equation}
(\nabla\dirsupp)^i=-\dirsupp^id\log{F}+\frac{\delta{}y^i}{F},
\label{nablaell}
\end{equation}
so that (\ref{deltaFdeltax=0}) leads to
\begin{equation}
\dirsupp^i_{\vert{}j}=0,
\qquad
\dirsupp^i_{\Vert{}j}=\delta^{i}_j-\dirsupp^i\dirsupp^{}_j,
\label{nablaellBis}
\end{equation}
for all $i,j=1,\ldots,n$.

Let us deduce from (\ref{nablaell}) a classical formula highlighting a special property of the Chern connection.
\begin{pro}
There holds
\begin{equation}
\omega^{\;a}_n=h^{\;a}_b\,\omega^\barb,
\label{omegana}
\end{equation}
where the $h_{ab}=\delta_{ab}-\delta^n_a\delta^n_b$, with $a,b=1,\ldots,n$, are the frame-components of the ``angular metric''.
\end{pro}

We end this section by a useful lemma (see Section 3.4 B in \cite{BCS} for a proof).
\begin{lem}\label{BianchiChern}
The first Bianchi identities 
imply
$A_{ij[k\Vert{}l]}=A_{ij[k}\dirsupp_{l]}$,
where the square brackets denote skew-symmetrization.
If we define the covariant derivative of the Cartan tensor (\ref{A}) in the direction, $\dirsupp$, of of the supporting element by
\begin{equation}
\dot{A}_{ijk}=A_{ijk|l}\,\dirsupp^l,
\label{dotA}
\end{equation}
then, the same Bianchi identities lead to
$P_{ijk}=\dirsupp^l{}P_{lijk}=-\dot{A}_{ijk}$.
\end{lem}


\subsubsection{The Cartan connection}

The Cartan connection is another prominent Finsler linear connection on $\pi^*(TM)$ which is related to the Chern connection in a simple way.
\begin{defi} \cite{BCS1,BCS}
Let $(M,F)$ be a Finsler structure with Cartan tensor~$(A_{abc})$, and let $(\omega_{ab})$ denote its canonical Chern connection. The frame components of the Cartan connection of $(M,F)$ are defined by
\begin{equation}
\homega_{ab}=\omega_{ab}+A_{abc}\,\omega^{\barc},
\label{omegaCartan}
\end{equation}
where $\barc=c+n$, with $a,b,c=1,\ldots,n$.
\end{defi}

The fundamental virtue of these connection $1$-forms is the skewsymmetry
\begin{equation}
\homega_{ab}+\homega_{ba}
=0,
\label{omegaCartanSkew}
\end{equation}
that guarantees that the fundamental tensor is parallel, $\hnabla\rg=0$, where $\hnabla$ stands for the covariant derivative associated with the Cartan connection. 
The Cartan connection is not symmetric; its torsion tensor $\hOmega^{a}=d\omega^a-\omega^b\wedge\homega_b^{a}$, is nonzero. Indeed, $\hOmega^{a}=\Omega^{a}-\omega^b\wedge{}A_{abc}\,\omega^{\barc}$, and, since $\Omega^a=0$, it retains the form
\begin{equation}
\hOmega^{a}=-A^a_{\;bc}\,\omega^b\wedge\omega^{\barc}.
\label{CartanTorsion}
\end{equation}
One easily proves the following result.
\begin{thm}
There exists a unique linear connection, $(\homega^a_b)$, on $\pi^*(TM)$ whose torsion, $(\hOmega^a)$, is given by (\ref{CartanTorsion}), and which is $\rg$-compatible, $\hnabla\rg=0$, as expressed by~(\ref{omegaCartanSkew}). This connection is the Cartan connection (\ref{omegaCartan}).
\end{thm}
Other characterizations of the Cartan connection 
can be found in the literature, e.g, in \cite{AP,AIM,BF}. 
\begin{pro}
The torsion of the Cartan connection is such that
\begin{equation}
\hOmega^n=0
\label{hOmegan=0}.
\end{equation}
\end{pro}

From now on, and whenever possible, we will use frame indices, $a,b,c,\ldots$, rather than local coordinate indices, $i,j,k,\ldots$.

\goodbreak


\begin{pro}\cite{BCS1}
The curvature $\hOmega^{\;a}_b=d\homega^{\;a}_b-\homega^{\;c}_b\wedge\homega^{\;a}_c$ of the Cartan connection is given by
\begin{equation}
\hOmega_{ab}=\Omega_{ab}+A_{abc}\,\Omega^{\;c}_n+A_{abd|c}\,\omega^c\wedge\omega^\bard+A^e_{\;ad}A^{}_{bce}\,\omega^\barc\wedge\omega^\bard
\label{hOmega}
\end{equation}
where $(\Omega^{\;a}_b)$ denotes the Chern curvature $2$-form (\ref{Omega}) with (\ref{R}, \ref{P}), and~$(A_{abc})$ the Cartan tensor~(\ref{A}). One has the decomposition
\begin{equation}
\hOmega^{\;a}_b=\half\hR^{\;a}_{b\;cd}\,\omega^c\wedge\omega^d+\hP^{\;a}_{b\;cd}\,\omega^c\wedge\omega^\bard+\half\hQ^{\;a}_{b\;cd}\,\omega^\barc\wedge\omega^\bard,
\label{hOmegaGeneral}
\end{equation}
with
\begin{eqnarray}
\label{hR}
\hR^{\;a}_{b\;cd}&=&R^{\;a}_{b\;cd}+A^{a}_{\;be}\,R^e_{\;cd},\\
\label{hP}
\hP^{\;a}_{b\;cd}&=&P^{\;a}_{b\;cd}+A^{a}_{\;bd|c}-A^a_{\;be}\dot{A}^e_{\;cd},\\
\label{hQ}
\hQ^{\;a}_{b\;cd}&=&2A^a_{\;e[c}A^e_{\;d]b},
\end{eqnarray}
where we use the notation $R^e_{\;cd}=R^{\;e}_{n\;cd}$, and $\dot{A}_{abc}=A_{abc|n}$ (see~(\ref{dotA})).
\end{pro}



\section{Geometrical optics in Finsler spaces}
\label{SectionGeometricalOpticsFinsler}

\subsection{Finsler geodesics}

Following Souriau's terminology \cite{Sou}, we call \textit{evolution space} the indicatrix-bundle 
\begin{equation}
SM=F^{-1}(1)
\label{SM}
\end{equation}
above $M$, as it actually hosts the dynamics given by a presymplectic structure; the latter will eventually be inherited from the Finsler metric on $\TMM$. 

Denote, again, by $\iota:SM\hookrightarrow\TMM$ the canonical embedding. The fundamental geometric object governing the geodesic spray on $SM$ is the $1$-form
\begin{equation}
\varpi=\iota^*\omega_H,
\label{varpi}
\end{equation}
i.e., the pull-back on $SM$ of Hilbert $1$-form $\omega_H$ (see~(\ref{Hilbert})). The direction of this $1$-form defines a contact structure on the $(2n-1)$-dimensional manifold $SM$, since $\varpi\wedge(d\varpi)^{n-1}\neq0$. See \cite{Fou,Che1,Che}. The following lemma is classical.
\begin{lem}
The exterior derivative of the Hilbert $1$-form is given by
\begin{equation}
d\omega_H=
\delta_{AB}\,\omega^\barA\wedge\omega^B
\label{domegaH}
\end{equation}
with $A,B=1,\ldots,n-1$.
\end{lem}
\begin{rmk}
{\rm
The exterior derivative of the Hilbert $1$-form is independent of the choice of a linear connection; it depends only on the non-linear connection~(\ref{Splitting}).
}
\end{rmk}
The main result regarding Finsler geodesics can be stated as follows. See also~\cite{Fou} for a full account on the geometry of second order differential equations.

\begin{thm}\label{ThmGeodesicSpray}
The geodesic spray of a Finsler structure $(M,F)$ is the vector field~$X$ of $SM$ uniquely defined by
\begin{equation}
\sigma(X)=0,
\qquad
\qquad
\varpi(X)=1,
\label{geodesicSpray}
\end{equation}
where $\sigma=d\varpi$.
\end{thm}
\begin{proof}
Write $X\in\Vect(SM)$ in the form $X=X^A\,\he_A+X^n\,\he_n+X^\barA\,\he_\barA+X^\barn\,\he_\barn$,  see~(\ref{orthonormalBases}), using dummy indices $A$ and $\barA=A+n$, where $A=1,\ldots,n-1$. Since $X\in\Vect(\TMM)$ is tangent to $SM$ iff $X(F)=\omega^\barn(X)=0$, as clear from (\ref{omegabn}), we have $X\in\ker(\sigma)$ iff $d\omega_H(X)+\lambda\omega^\barn=0$ where $\lambda\in\bbR$ is a Lagrange multiplier. The latter equation readily yields, with the help of~(\ref{domegaH}), $\delta_{AB}(X^\barA\,\omega^B-X^B\,\omega^\barA)+\lambda\omega^\barn=0$, hence $X^A=X^\barA=0$, and $\lambda=0$. Then $X=X^n\,\he_n+X^\barn\,\he_\barn$ is actually tangent to~$SM$ if $\omega^\barn(X)=X^\barn=0$, which leads to 
\begin{equation}
X\in\ker(\sigma)
\qquad
\Longleftrightarrow
\qquad
X=X^n\he_n
\label{kersigma0}
\end{equation}
for some $X^n\in\bbR$. Thus, $(SM,\sigma)$ is a presymplectic manifold. The quotient $SM/\ker(\sigma)$ is the set of oriented Finsler geodesics, which (if endowed with a smooth structure) becomes a $(2n-2)$-dimensional symplectic manifold, see~\cite{Alv}.

We then find that $\varpi(X)=X^n$, and the constraints (\ref{geodesicSpray}) express the fact that~$X$ is the Reeb vector field, and retains the form $X=\he_n$, which, in view of (\ref{en}), we can write
\begin{equation}
X=\dirsupp^i\frac{\delta}{\delta{}x^i}.
\label{geodesicSprayBis}
\end{equation}
The vector field (\ref{geodesicSprayBis}) of $SM$ is the \textit{geodesic spray} \cite{BCS} of the Finsler structure.
\end{proof}

\subsection{Geometrical optics in anisotropic media}

The geodesic spray, $X$, given by (\ref{geodesicSprayBis}), integrates to a Finsler geodesic flow, $\varphi_t$, on the bundle $SM$ via the ordinary differential equation 
$
d\varphi_t(x,\dirsupp)/dt=X(\varphi_t((x,\dirsupp))
$
for all~$t\in{}I\subset\bbR$. The latter translates as
\begin{equation}
\left\{
\begin{array}{lcl}
\displaystyle
\frac{dx^i}{dt}&=&\dirsupp^i\\[8pt]
\displaystyle
\frac{d\dirsupp^i}{dt}&=&-G^i(x,\dirsupp)
\end{array}
\right.
\label{dotx=udotu=-G}
\end{equation}
where the acceleration components (or spray coefficients) read
$G^i=N^i_{\;j}y^j$ (see \cite{BCS}), for $i=1,\ldots,n$. 
The geodesic flow then defines geodesics per se, $x_t=\pi(\varphi_t(x,\dirsupp))$, of the base manifold,~$M$, with initial data $(x,\dirsupp)\in{}SM$. 

\subsubsection{The Fermat Principle}\label{FermatSection}

\begin{defi}\cite{Kne}
Two Finsler structures  $(M,F)$ and $(M,\widetilde{F})$ are said to be conformal\-ly related if $\widetilde{F}(x,y)=\sfn(x)F(x,y)$ for some $\sfn\in{}C^\infty(M,\bbR^*_+)$.
\end{defi}

If $F$ is a Riemannian structure, then $\widetilde{F}$ is a Riemannian structure conformal\-ly related to $F$, since their metric tensors are such that $\widetilde{\rg}_{ij}(x)=\sfn^2(x)\,\rg_{ij}(x)$. In this case, the geodesics of $(M,\widetilde{F})$ may be interpreted as the trajectories of light in a medium, modeled on the Riemannian manifold $(M,F)$, and endowed with a refractive index $\sfn$. This is, in essence, the \textit{Fermat Principle} of geometrical optics. 

\begin{pro}\label{proFermat}
Let $(M,F)$ and $(M,\widetilde{F})$ be conformally related Finsler structures, i.e., be such that $\widetilde{F}(x,y)=\sfn(x)F(x,y)$ for a given $\sfn\in{}C^\infty(M,\bbR^*_+)$, called their relative refractive index. Their geodesic sprays are related as follows:
\begin{equation}
X=\dirsupp^i\frac{\delta}{\delta{x^i}}
\qquad
\&
\qquad
\widetilde{X}=\frac{1}{\sfn}\dirsupp^i\frac{\delta}{\delta{x^i}}+\frac{1}{\sfn^3}(\rg^{ij}-2\dirsupp^i\dirsupp^j)\frac{\partial\sfn}{\partial{}x^j}\,\frac{\partial}{\partial y^i},
\label{ConfGeodesicSpray}
\end{equation}
where $\dirsupp^i=y^i/F$, for $i=1,\ldots,n$. Putting $\dx^i=\sfn\widetilde{X}(x^i)$, and $\dy^i=\sfn\widetilde{X}(y^i)$, we obtain the equations of the geodesics of $(M,\widetilde{F})$ in the following guise:
\begin{equation}
\left\{
\begin{array}{rcl}
\dx^i&=&\dirsupp^i\\[6pt]
\nabla_\dirsupp(\sfn\,\dirsupp)^i&=&\displaystyle\rg^{ij}\frac{\partial\sfn}{\partial{}x^j}
\end{array}
\right.
\label{FinslerFermatEqs}
\end{equation}
where $\nabla_\dirsupp$ is the covariant derivative with reference vector $u$, defined, for all vector field, $v$, along the curve with velocity $u$, by  $\nabla_\dirsupp(v)^i=\dv^i+\Gamma^i_{jk}(x,\dirsupp)\dirsupp^jv^k$.
\end{pro}
\begin{proof}
The Hilbert $1$-forms are related by $\widetilde{\omega}_H=\sfn\,\omega_H$, and their exterior derivatives by $\widetilde\sigma=\sfn\,\sigma+d\sfn\wedge\omega_H$. In other words $\widetilde\sigma=\sfn\,\delta_{AB}\,\omega^\barA\wedge\omega^B+ \sfn_A\,\omega^A\wedge\omega^n$, where $\sfn_A=e^i_A\,\partial_i\sfn$. 

Reproducing the proof of Theorem \ref{ThmGeodesicSpray}, we will decompose $\widetilde{X}\in\Vect\big(\widetilde{S}M\big)$ as $\widetilde{X}=\widetilde{X}^A\,\he_A+\widetilde{X}^n\,\he_n+\widetilde{X}^\barA\,\he_\barA+\widetilde{X}^\barn\,\he_\barn$,  with the same notation as before. Again $\widetilde{X}\in\ker(\widetilde{\sigma})$ iff $d\widetilde{\omega}_H(\widetilde{X})+\lambda\widetilde{\omega}^\barn=0$ for some $\lambda\in\bbR$. (Note that, in view of (\ref{omegabn}), we have $\widetilde{\omega}^\barn=\omega^\barn+d\sfn/\sfn$.) This equation readily leaves us with $\widetilde{X}^A=0$, and $\widetilde{X}^\barA=\sfn^A\,\widetilde{X}^n$, for all $A=1,\ldots,n-1$, together with $\lambda=0$. 

At last, $\widetilde{X}$ is tangent to $\widetilde{S}M$ if $\widetilde{\omega}^\barn(\widetilde{X})=0$, i.e., if $\widetilde{X}^\barn=-(\sfn^n/\sfn)\widetilde{X}^n$. Then, $\widetilde{X}$ is the Reeb vector field for $\widetilde{F}$ if $\widetilde{\omega}_H(\widetilde{X})=1$, i.e., if $\widetilde{X}^n=1/\sfn$. The geodesic spray of the Finsler structure $(M,\widetilde{F})$ is thus
$$
\widetilde{X}=\frac{1}{\sfn}\left[\he_n+\frac{\sfn^A}{\sfn}\,\he_\barA-\frac{\sfn^n}{\sfn}\,\he_\barn\right]
$$
while that of the Finsler structure $(M,F)$ reduces to $X=\he_n$ by letting $\sfn=1$.

We thus recover (\ref{ConfGeodesicSpray}) via (\ref{orthonormalBases}) and (\ref{eni}), and also by the following fact, viz., $\sfn^A\,\he_\barA-\sfn^n\,\he_\barn=(\delta^{AB}e^i_A e^j_B-e^i_n e^j_n)\partial_j\sfn\,F\partial/\partial{y^i}=(\rg^{ij}-2\dirsupp^i\dirsupp^j)(\partial_j\sfn/\sfn)\partial/\partial{y^i}$, since~$\widetilde{F}=\sfn\,F=1$ on $\widetilde{S}M$.

Let us now derive Equations (\ref{FinslerFermatEqs}); we first notice that $\dx^i=\sfn\widetilde{X}(x^i)=\dirsupp^i$, and also that $\dy^i=\sfn\widetilde{X}(y^i)=-N^i_j\,\dirsupp^j+(F/\sfn)(\rg^{ij}-2\dirsupp^i\dirsupp^j)\partial_j\sfn$. Defining, as in, e.g.,~\cite{BCS}, the covariant derivative of the vector field $v$, with reference vector $u$, by the expression $\nabla_\dirsupp(v)^i=\dv^i+\Gamma^i_{jk}\dirsupp^jv^k=\dv^i+N^i_j\,v^j/F$, see (\ref{NijChern}), enables us to compute the ``geodesic acceleration'' $\nabla_\dirsupp(\sfn\,\dirsupp)$. Since $(\sfn\,\dirsupp)^i=(\sfn^2\,y)^i$, we get
$\nabla_\dirsupp(\sfn\,\dirsupp)^i=2\sfn\dot{\sfn}\,y^i+\sfn^2\nabla_\dirsupp(y)^i=2\sfn\,\dirsupp^j\partial_j\sfn\,y^i+\sfn^2(\dy^i+N^i_j\,\dirsupp^j)=2\sfn\,\dirsupp^j\partial_j\sfn\,y^i+\sfn^2(F/\sfn)(\rg^{ij}-2\dirsupp^i\dirsupp^j)\partial_j\sfn=\rg^{ij}\partial_j\sfn$, and we are done.
\end{proof}

\begin{rmk}
The differential equations (\ref{FinslerFermatEqs}) generalize, to the Finsler framework, the \textit{Fermat equations} ruling the propagation of light in a Riemannian manifold, through a dielectric medium of refractive index $\sfn$.
\end{rmk}

\subsubsection{Finsler optics}

It has originally been envisioned by Finsler himself (see, e.g., \cite{AIM,Ing}) that the \textit{indicatrix} $S_xM=\{\dirsupp\in{}T_xM\,\vert\,F(x,\dirsupp)=1\}$ of a Finsler structure $(M,F)$ might serve as a model for the geometric locus of the ``phase velocity'' of light waves at a point~$x\in{}M$. The fact that, in anisotropic optical media, the velocity of a (plane) light-wave specifically depends upon the direction of its propagation, prompted him to put forward a classical (as opposed to field-theoretical) model of geometrical optics in anisotropic, non dispersive, media ruled by Finsler structures. Finsler geodesics have therefore consistently received the interpretation of light trajectories in such optical media. Let us mention, among many an example, an application of Finsler optics to dynamical systems engendered by Finsler billiards \cite{GT}.

When specialized to Riemannian structures, e.g., to Fermat structures presented in Section \ref{FermatSection}, Finsler geodesics are nothing but plain Riemannian geodesics, regarded as light rays in (non homogeneous) isotropic media. See, e.g., \cite{CN1,CN2,DHH,DHH2}.

This justifies the following principle of Finsler geometrical optics.

\begin{defi}
The light rays in a non-homogeneous, anisotropic, optical medium described by a Finsler structure $(M,F)$ are the oriented geodesics associated with the geodesic spray (\ref{geodesicSprayBis}) of this Finsler structure. 
\end{defi}

\goodbreak

For example, \textit{birefringent} media (solid or liquid crystalline media) can be described by a pair of Finsler metrics, namely, the \textit{ordinary} (resp. \textit{extra\-ordinary}) metric $F_o$ (resp. $F_e$) attached to a (three-dimensional) manifold, $M$, representing the anisotropic optical medium. Those are respectively given, in the particular case of \textit{uniaxial} crystals, in terms of a pair of Riemannian metrics $a=a_{ij}(x)\,dx^i\otimes{}dx^j$, and $b=b_{ij}(x)\,dx^i\otimes{}dx^j$ on $M$, by \cite{AIM,Ing}
\begin{eqnarray}
\label{Fo}
F_o(x,y)&=&\sqrt{a_{ij}(x)\,y^iy^j},\\[6pt]
\label{Fe}
F_e(x,y)&=&\frac{a_{ij}(x)\,y^iy^j}{\sqrt{b_{ij}(x)\,y^iy^j}}.
\end{eqnarray}
The geodesics of the metric $F_e$ are meant to describe extra\-ordinary light rays, whereas those of the Riemannian metric, $F_o$, will merely lead to ordinary rays. 

\begin{rmk}\label{Rmkab}
{\rm
Let us emphasize that $(M,F_e)$, where $F_e$ is as in (\ref{Fe}), is a Finsler structure if its fundamental tensor
\begin{equation}
\rg^e_{ij}=F_e^2\left[\frac{2\,a_{ij}}{a(y,y)}-\frac{b_{ij}}{b(y,y)}\right]
+\frac{4\,c_i\,c_j}{b(y,y)^3},
\label{geij}
\end{equation}
where $c_i=a(y,y)\,b_{ij}(x)y^j-b(y,y)\,a_{ij}(x)y^j$, is positive definite. This is, indeed, the case if the quadratic forms $a$ and $b$ 
verify 
$b/{\sqrt{2}}<a<b\,\sqrt{2}$,
everywhere on~$\TMM$.
}
\end{rmk}

The more complex case of \textit{biaxial} optical media is also studied in \cite{AIM}, and gives rise to a pair of specific Finsler metrics
\begin{equation}
F^\pm(x,y)=\frac{a_{ij}(x)\,y^iy^j}{\sqrt{b^\pm_{ij}(x)\,y^iy^j}}
\label{Fpm}
\end{equation}
where $a,b^+$, and $b^-$ are Riemannian metrics characterizing the optical properties of the anisotropic medium. (Let us note that Remark \ref{Rmkab} applies just as well for the metrics (\ref{Fpm})).

\subsubsection{The example of birefringent solid crystals}

Let us review how Finsler metrics of the form (\ref{Fpm}), or (\ref{Fo}) and (\ref{Fe}), arise in the particular case of anisotropic solid crystals. To that end, we revisit the original derivation \cite{AIM} of the Minkowski norms that account for the propagation of light in anisotropic dielectric solids with principal (positive) velocities $v_1,v_2,v_3$.

\goodbreak

In the framework of Maxwell's wave optics, the Fresnel equation of wave normals
$$
\hu_1^2(\Vert\bu\Vert^2-v_2^2)(\Vert\bu\Vert^2-v_3^2)+
\hu_2^2(\Vert\bu\Vert^2-v_3^2)(\Vert\bu\Vert^2-v_1^2)+
\hu_3^2(\Vert\bu\Vert^2-v_1^2)(\Vert\bu\Vert^2-v_2^2)=0
$$
expresses the dependence of the phase velocity, $\bu$, of a plane wave upon its direction of propagation, $\bhu=\bu/\Vert\bu\Vert$, in such a medium; we denote, here, by $\Vert\,\cdot\,\Vert$ the norm on standard Euclidean space $(\bbR^3, \la\,\cdot\,,\,\cdot\,\ra)$.

- Assuming, e.g., $v_1>v_2>v_3$, one solves the Fresnel equation for the norm of the phase velocity, viz., $\Vert\bu\Vert^2=A+B\cos(\theta'\pm\theta'')$ where $\theta'$ and $\theta''$ are the angles between the direction of propagation, $\bhu$, and the (oriented) optical axes 
$\be'$ and $\be''$;
the scalars $A=\half(v_1^2+v_3^2)$, and $B=\half(v_1^2-v_3^2)$, as well as the vectors $\be'$, and $\be''$, are 
characteristic of the crystal \cite{BW}. 
The Minkowski norm, $F$, associated with each
solution of the Fresnel equation is easily found \cite{AIM} using Okubo's trick that amounts to the replacement $\bu\rightsquigarrow\by/F(\by)$, insuring that $F(\bu)=1$. Easy calculation leads us to $\Vert\bu\Vert^2=\Vert\by\Vert^2/F(\by)^2=A+B\Vert\by\Vert^{-2}\left[\la\be',\by\ra\la\be'',\by\ra\mp\Vert\be'\times\by\Vert\Vert\be''\times\by\Vert\right]$, that is, to
\begin{equation}
F^\pm(\by)=\frac{\Vert\by\Vert^2}{\sqrt{A\Vert\by\Vert^2{\strut}+B\Big[\la\be',\by\ra\la\be'',\by\ra\mp\Vert\be'\times\by\Vert\Vert\be''\times\by\Vert\Big]}},
\label{Fpmbis}
\end{equation}
where $\times$ denotes the standard Euclidean cross-product. This expression admits straightforward generalizations to the case of fluid crystals, Faraday-active media, etc., where the quantities $A$, $B$, $\be'$, and $\be''$ become position-dependent; it ultimately leads to the expression (\ref{Fpm}) of a pair of Finsler metric for general biaxial media.

- The case of uniaxial media is treated by assuming, e.g., $v_1=v_2>v_3$, which 
implies $\be'=\be''(=\be)$.
The Minkowski norms (\ref{Fpmbis}) admit a prolongation to this situation and read
\begin{eqnarray}
\label{Fobis}
F^-(\by)&=&\frac{\Vert\by\Vert}{v_1},\\
\label{Febis}
F^+(\by)&=&\frac{\Vert\by\Vert^2}{\sqrt{v_3^2\Vert\by\Vert^2{\strut}+(v_1^2-v_3^2)\la\be,\by\ra^2}}.
\end{eqnarray}
Those correspond, respectively, to an ordinary Euclidean metric, $F_o=F^-$, and to an extraordinary Minkowski metric, $F_e=F^+$, again generalized by~(\ref{Fo}), and~(\ref{Fe}).

- The last case, for which $v_1=v_2=v_3$, clearly leads to a single Euclidean metric, namely $F(\by)=\Vert\by\Vert/v_1$, that rules geometrical optics in isotropic media with refractive index $\sfn=1/v_1$.


\section{Geometrical spinoptics in Finsler spaces}
\label{SectionFinslerSpinoptics}

So far, the polarization of light has been neglected in the various formulations of geometrical optics. We contend that spinning light rays do, indeed, admit a clear cut geometrical status allowing for a natural extension of plain geometrical optics to the case of circularly polarized light rays (Euclidean photons) traveling in arbitrary non dispersive optical media. 

The touchstone of our viewpoint about geometrical optics for spinning light is the \textit{Euclidean symmetry} of the manifold of oriented lines in (flat) Euclidean space. This fundamental symmetry will be taken as a guiding principle to set up a model that could describe the geometry of spinning light rays in quite general, crystalline and liquid, optical media. See~\cite{DHH} and~\cite{DHH2} for a first approach to geometrical \textit{spinoptics} in inhomogeneous, isotropic, media.


We will therefore start by some elementary facts about the symplectic structure of the space of oriented lines in Euclidean space. The consideration of the generic coadjoint orbits of the Euclidean group will then be justified on physical grounds. 

Let us recall that, if we denote by $\Ad$ the adjoint action of a Lie group, and by $\Coad$, its coadjoint action, then the orbits of the latter action inherit a canonical structure of symplectic manifolds. These homogeneous symplectic manifolds play a central r\^ole in mechanics and physics, where some of them may be interpreted as the elementary systems associated with the symmetry group under consideration~\cite{Sou}. 

The following construction is standard.

\goodbreak

\begin{thm}\label{ThmCoadOrbits}
Let $G$ be a (finite-dimensional) Lie group $G$ with Lie algebra~$\fg$. Fix $\mu_0\in\fg^*$ and define the following $1$-form 
\begin{equation}
\varpi_{\mu_0}=\mu_0\cdot\vartheta_G,
\label{canonical1form}
\end{equation}
where $\vartheta_G$ is the left-invariant Maurer-Cartan $1$-form of $G$. Then, $\sigma_{\mu_0}=d\varpi_{\mu_0}$ is a presymplectic $2$-form on $G$ which 
is the pull-back of
the canonical Kirillov-Kostant-Souriau 
symplectic $2$-form on the $G$-coadjoint orbit 
\begin{equation}
\cO_{\mu_0}=\{\mu=\Coad_g(\mu_0)\,\vert\,g\in{}G\}\cong{}G/G_{\mu_0},
\end{equation}
where $G_{\mu_0}$ is the stabilizer of $\mu_0\in\fg^*$.
\end{thm}

\goodbreak

\subsection{Spinoptics and the Euclidean group}

From now on we will confine considerations to three-dimensional configuration spaces to comply with the physical principles of geometrical optics.

\goodbreak

An oriented straight line, $\xi$, in Euclidean affine space $(E^3, \la\,\cdot\,,\,\cdot\,\ra)$ is determined by its direction, a vector $\bdirsupp\in\bbR^3$ of unit length, and an arbitrary point $Q\in\xi$. Having chosen an origin, $O\in{}E^3$, we may consider the vector $\bq=Q-O$, orthogonal to~$\bdirsupp $. The set of oriented, non parametrized, straight lines is thus the smooth manifold
\begin{equation}
\cM=\{\xi=(\bq,\bdirsupp)\in\bbR^3\times\bbR^3\,\vert\,\la\bdirsupp,\bdirsupp\ra=1,\la \bdirsupp,\bq\ra=0\},
\label{StraightLines}
\end{equation}
i.e., the tangent bundle $\cM\cong{}TS^2$ of the round sphere $S^2\subset\bbR^3$,
which has been recognized by Souriau \cite{Sou} as a coadjoint orbit of the group, $\rE(3)$, of Euclidean isometries, and inherits, as such, an $\rE(3)$-invariant symplectic structure. See also~\cite{Igl}. 

Consider the group, $\SE(3)=\SO(3)\ltimes\bbR^3$, of orientation-preserving Euclidean isometries of $(E^3, \la\,\cdot\,,\,\cdot\,\ra,\bvol)$, viewed as the matrix-group whose elements read 
\begin{equation}
g=\pmatrix{R&\bx\cr0&1},
\label{SE3}
\end{equation}
where $R\in\SO(3)$, and $\bx\in\bbR^3$.

The (left-invariant) Maurer-Cartan $1$-form of $\SE(3)$ is therefore given by
\begin{equation}
\vartheta_{\SE(3)}=\pmatrix{\bhomega&\bomega\cr0&0},
\label{MCE3}
\end{equation}
where $\bhomega=R^{-1}dR$, and $\bomega=R^{-1}d\bx$.

Let $\mu=(\bS,\bP)$ denote a point in $\se(3)^*$ where $\se(3)=\so(3)\ltimes\bbR^3$ is the Lie algebra of~$\SE(3)$. We will use the identification $\so(3)\cong\Lambda^2\bbR^3$ (resp. $(\bbR^3)^*\cong\bbR^3$) given by $S^{\;b}_a=S^{ab}$ (resp. $P_a=P^a$), where
\begin{equation}
S^{ab}+S^{ba}=0,
\label{Sab+Sba=0}
\end{equation}
for all $a,b=1,2,3$. The pairing $\se(3)^*\times\se(3)\to\bbR$ will be defined by 
\begin{eqnarray}
\label{paringe3}
(\bS,\bP)\cdot(\bhomega,\bomega)&=&-\half\Tr(\bS\,\bhomega)+\la\bP,\bomega\ra\\[6pt]
\label{pairinge3Bis}
&=&\half\,S^{ab}\,\homega_{ab}+P_a\,\omega^a.
\end{eqnarray}
The coadjoint representation of $\SE(3)$, viz., $\Coad_g\mu\equiv\mu\circ\Ad_{g^{-1}}$, is given by $\Coad_{g}(\bS,\bP)=(R(\bS+\bx\wedge\bP)R^{-1},R\bP)$.
Clearly, $C=\la\bP,\bP\ra$ and $C'=(\bS\wedge\bP)/\bvol$ are co\-adjoint $\SE(3)$-invariants. These are the only invariants of the Euclidean co\-adjoint representation, and fixing~$(C,C')$ or $(C=0,C'')$, where $C''=\half{}S^{ab}S_{ab}$, yields a single coadjoint orbit \cite{GS,MR,Sou}.

\subsubsection{Colored light rays}

Specializing the construction of Theorem \ref{ThmCoadOrbits} to the case $G=\SE(3)$, with $C=p^2$ and $p>0$ together with $C'=0$, we can choose $\mu_0=(\mathbf{0},\bP_0)$ and $\bP_0=(0,0,p)$. The invariant $p$ is the \textit{color}  \cite{Sou} of the chosen coadjoint orbit.

The $1$-form (\ref{canonical1form}) then associated, via the pairing (\ref{pairinge3Bis}), and the Maurer-Cartan $1$-form~(\ref{MCE3}), with the invariant $p$ is thus 
$\varpi_{\mu_0}=\la\bP_0,\bomega\ra$, or
\begin{eqnarray}
\label{varpip0}
\varpi_{\mu_0}&=&P_a\,\omega^a\\
\label{varpip0Bis}
&=&p\,\omega^3.
\end{eqnarray}
Straightforward calculation yields $\varpi_{\mu_0}=p\,\dirsupp_idx^i$,
with $\dirsupp_i=\delta_{ij}\,\dirsupp^j$, where $\bdirsupp=\be_3$ is the third vector of the orthonormal, positively oriented, basis $R=(\be_1,\be_2,\be_3)$. 

The $1$-form (\ref{varpip0Bis}) is, up to an overall multiplicative constant, $p$, equal to the canonical $1$-form~(\ref{varpi}) on the sphere-bundle $SE^3$, associated with the trivial Finsler structure $(E^3,F)$, with $F(\bx,\by)=\sqrt{\la\by,\by\ra}$. Proposition \ref{proFermat} just applies, with $n=1$, and conforms to Euclid's statement that light, whatever its color, travels in vacuum along oriented geodesics of $E^3$.
Indeed, the exterior derivative of the $1$-form $\varpi_{\mu_0}$ of $SE^3$ reads
$
\sigma_{\mu_0}=p\,d\dirsupp_i\wedge{}dx^i.
$
Its kernel, given by
\begin{equation}
X\in\ker(\sigma_{\mu_0})
\qquad
\Longleftrightarrow
\qquad
X=\lambda\,\dirsupp^i\frac{\partial}{\partial{}x^i},
\label{kersigmap0}
\end{equation}
with $\lambda\in\bbR$, yields the (flat Euclidean) geodesic spray $(\lambda=1)$. We will resort to generalizations of this particular construct of the geodesic foliation.

The integral invariant, $\sigma_{\mu_0}$, 
descends, as a symplectic $2$-form, to the quotient $\cM=SE^3/\ker(\sigma_{\mu_0})$ described by $\xi=(\bq,\bdirsupp)$, where $\bq=\bx-\bdirsupp\la\bdirsupp,\bx\ra$. This is the content of Theorem~\ref{ThmCoadOrbits} insuring that $\cM\cong{}TS^2\subset\se(3)^*$ is endowed with a canonical symplectic structure, namely~$(\cM,p\,d\dirsupp_i\wedge{}dq^i)$. 

We note that the $\SE(3)$-coadjoint orbit $\cO_{\mu_0}$ is an $\rE(3)$-coadjoint orbit.

\subsubsection{The spinning and colored Euclidean coadjoint orbits}

The \textit{generic} $\SE(3)$-coadjoint orbits are, in fact, characterized by the \textit{Casimir invariants} $C=p^2$, with $p>0$ (color), and $C'=s p$ where $s\neq0$ stands for \textit{spin}. We call \textit{helicity} the sign of the spin invariant, $\varepsilon=\sign(s)$.


The origin, $\mu_0$, of such an orbit can be freely chosen so as to satisfy the constraints
$S^{ab}P_b=0$, for all $a=1,2,3$, and $\half\,S^{ab}\,S_{ab}=s^2$, together with $P_aP^a=p^2$. One may posit
\begin{equation}
\bS_0=\pmatrix{
0&-s&0\cr
s&0&0\cr
0&0&0
}
\qquad
\hbox{and}
\qquad
\bP_0=(0,0,p).
\label{S0P0}
\end{equation}
The coadjoint orbit, $\cO_{\mu_0}$, passing through $\mu_0=(\bS_0,\bP_0)\in\se(3)^*$  is, again, diffeomorphic to $TS^2$. It is endowed with the symplectic structure coming from the $1$-form~(\ref{canonical1form}) on the group~$\SE(3)$, which now reads
\begin{eqnarray}
\label{varpips}
\varpi_{\mu_0}&=&P_a\,\omega ^a+\half\,S^{ab}\,\homega_{ab}\\[6pt]
\label{varpipsBis}
&=&p\,\omega^3+s\,\homega_{12},
\end{eqnarray}
where $\bhomega$ (resp. $\bomega$) stand for the flat Levi-Civita connection (resp. soldering) $1$-form on the bundle, $\SE(3)\cong\SO(E^3)$, of positively oriented, orthonormal frames of $E^3$.

\begin{rmk}
{\rm
The $1$-form (\ref{varpips}) is the central geometric object of the present study.
}
\end{rmk}


\begin{figure}[!hbp]
\begin{displaymath}
\xymatrix{
 & \SE(3) \ar[dl] \ar[d] &\\
\cO_{\mu_0}\cong{}TS^2 & \ar[l] S\bbR^3 \ar[r] & \bbR^3}
\end{displaymath}
\caption{\label{EuclidPhoton}}
\end{figure}

Straightforward computation then leads to 
$\varpi_{\mu_0}=p\la\be_3,d\bx\ra -s\la\be_1,d\be_2\ra$. 
The exterior derivative of $\varpi_{\mu_0}$ is found \cite{Sou,GS,MR,DHH} to be given by
\begin{equation}
\sigma_{\mu_0}=p\,d\dirsupp_i\wedge{}dx^i-\frac{s}{2}\epsilon_{ijk}\,\dirsupp^i d\dirsupp^j\wedge{}d\dirsupp^k,
\label{sigmasp}
\end{equation}
where, again, $\bdirsupp=\be_3$, and also $\epsilon_{ijk}$ stands for the signature of the permutation $\{1,2,3\}\mapsto\{i,j,k\}$. This $2$-form conspicuously descends to $S\bbR^3\cong\bbR^3\times{}S^2$. Its characteristic foliation is, \textit{verbatim}, given by (\ref{kersigmap0}): spinning light rays in vacuum are nothing but oriented Euclidean geodesics. As shown in the sequel, things will change dramatically for such light rays in a refractive medium.

\begin{rmk}
{\rm
Actually, ``photons'' are characterized by $|s|=\hbar$, where~$\hbar$ is the reduced Planck constant; \textit{right-handed} photons correspond to $s=+\hbar$, and \textit{left-handed} ones to $s=-\hbar$, see \cite{Sou}. We will, nevertheless, leave the parameter $s$ arbitrary when dealing with ``spinning light rays''.
}
\end{rmk} 

The manifold of spinning light rays, $\cO_{\mu_0}=SE^3/\ker(\sigma_{\mu_0})\cong{}TS^2$ (see Fig.~\ref{EuclidPhoton}) is, just as before, parametrized by the pairs $\xi=(\bq,\bu)$ and endowed with the ``twisted'' symplectic $2$-form
$\omega_{\mu_0}=p\,d\dirsupp_i\wedge{}dq^i-\frac{s}{2}\epsilon_{ijk}\,\dirsupp^i d\dirsupp^j\wedge{}d\dirsupp^k$.

Note that the union of two $\SE(3)$-coadjoint orbits defined by the invariants $(p,s)$ and $(p,-s)$ is symplecto\-morphic to a single $\rE(3)$-coadjoint orbit.

\begin{rmk}
{\rm
The $\SE(3)$-coadjoint orbits of spin~$s$, and color $p$, are symplectomorphic to Marsden-Weinstein reduced massless $\SE(3,1)_0$-coadjoint orbits of spin~$s$, at given (positive) energy $E=pc$, where $c$ stands for the speed of light in vacuum, see~\cite{DHH}. This justifies that the color, $p$, of Euclidean light rays corresponds, via reduction, to the energy of relativistic photons, or to the frequency of their associated mono\-chromatic plane waves \cite{Sou,GS,DHH2}.
}
\end{rmk}

\goodbreak

\subsection{Spinoptics in Finsler-Cartan spaces}

With these preparations, we formulate the principles of geometrical spin\-optics, with the premise that (i) Finsler structures should be considered a privileged geometric background for the description of  inhomogeneous, anisotropic, optical media, (ii) the original Euclidean symmetry which pervades geometrical optics should be invoked as a guiding principle in any formulation of spin extensions of geometrical optics. 

\subsubsection{Minimal coupling to the Cartan connection}


\begin{ax}\label{AxiomSpinoptics}
The trajectories of (circularly) polarized light, originating from an Euclidean co\-adjoint orbit $\cO_{\mu_0}\subset\se^*(3)$ with color $p>0$, and spin $s\neq0$, in a three-dimensional Finsler manifold~$(M,F)$, are governed by the following $1$-form on the principal bundle $\SO_2(SM)$ over evolution space $SM=F^{-1}(1)$, namely
\begin{equation}
\varpi_{\mu_0}=\mu_0\cdot\vartheta,
\label{omegaspFinsler0}
\end{equation}
where $\vartheta=(\bhomega,\bomega)$ is the ``affine'' Cartan connection, $\bhomega=(\homega^{\;a}_b)$ denoting the \textit{Cartan connection} (\ref{omegaCartan}), and $\bomega=(\omega^a)$ the coframe (\ref{omegaaLoc}). The characteristic foliation of the $2$-form $\sigma_{\mu_0}=d\varpi_{\mu_0}$ yields the differential equations of spinoptics in a medium described by the considered Finsler structure.
\end{ax}

\goodbreak

The $1$-form (\ref{omegaspFinsler0}) corresponds, \textit{mutatis mutandis}, to the Euclidean $1$-form (\ref{varpips}). (In (\ref{omegaspFinsler0}), and from now on, we simplify the notation and denote by $\vartheta$ the pull-back~$\iota^*\vartheta$ on $\SO_2(SM)$ of the corresponding $1$-form of $\SO_2(\TMM)$.) In Axiom~\ref{AxiomSpinoptics}, the replacement of the Euclidean group $\SE(3)$, see Fig. \ref{EuclidPhoton}, by the principal bundle~$\SO_2(SM)$, see Fig. \ref{FinslerPhoton}, and of the Maurer-Cartan $1$-form by the affine Cartan connection is akin to the so-called \textit{procedure of minimal coupling}. We refer to \cite{Kun} where the minimal coupling of a spinning particle to the gravitational field was originally introduced in the general relativistic framework. See also \cite{Sou2,Ste}.

\begin{figure}[!hbp]
\begin{displaymath}
\xymatrix{
\SO_2(SM) \ar[d] & &\\
SM \ar[r]^{\iota\ \ }& \TMM \ar[r]^{\ \ \pi} & M \\
}
\end{displaymath}
\caption{\label{FinslerPhoton}}
\end{figure}

The choice of the Cartan connection is impelled by the fact that the group underlying Finsler-Cartan geometry is the Euclidean group, $\rE(n)$, which is precisely the fundamental symmetry group of the symplectic model of free photons, for $n=3$. 
Indeed, the ``flat'' $n$-dimensional Finsler-Cartan structure defined by both conditions $\Omega^a=0$, and $\hOmega^{\;a}_b=0$, for all $a,b=1,\ldots,n$, is (locally) iso\-morphic to the Euclidean space, $(E^n,\la\,\cdot\,,\;\cdot\,\ra)$: torsionfreeness yields $A_{abc}=0$, see~(\ref{CartanTorsion}), hence that $(M,F)$ is Riemannian; zero curvature then entails local flatness, via~(\ref{hOmegaGeneral}), and~(\ref{R}). The Euclidean group, $\rE(n)$, is then the group of automorphisms of the flat structure. 
Let us emphasize that the Cartan connection has been originally referred to as the ``connexion euclidienne'' in \cite{Car}.

The $1$-form (\ref{omegaspFinsler0}) of $\SO_2(SM)$ thus reads
\begin{equation}
\varpi_{\mu_0}=P_a\,\omega ^a+\half\,S^{ab}\,\homega_{ab}.
\label{omegaspFinsler}
\end{equation}
In view of the choice (\ref{S0P0}) of the moment $\mu_0=(\bS,\bP)\in\se(3)^*$, viz., 
\begin{equation}
S^{ab}=s\,\epsilon^{AB}\delta^{a}_A\delta^{b}_B
\qquad
\hbox{and}
\qquad
P_a=p\,\delta^{3}_a,
\label{mu0}
\end{equation}
for all $a,b=1,2,3$, where $\epsilon^{AB}=2\delta^{[A}_1\delta^{B]}_2$, for all $A,B=1,2$, we find
\begin{equation}
\varpi_{\mu_0}=p\,\omega^3+s\,\homega_{12}.
\label{omegaspBisFinsler}
\end{equation}

\begin{rmk}\label{RemarkBerry}
{\rm
The $1$-form $\varpi_{\mu_0}$ differs from the Hilbert $1$-form, $\omega_H=\omega^3$, by a spin-term, $\homega_{12}$. This term, canonically associated to a generic Euclidean coadjoint orbit, is new in the framework of Finsler geometry, and akin to the Berry connection~\cite{Ber}. See, e.g., \cite{Bli,OMN}. 
Putting, in (\ref{omegaspBisFinsler}), $s=\varepsilon\hbar$ for photons (where $\varepsilon=\pm1$ is \textit{helicity}), and $p=\hbar{}k$, where $k=\varepsilon/\lambdabar$ is the wave number, we observe that $s/p=\lambdabar$, so that Formula~(\ref{omegaspBisFinsler}) indeed corresponds to~(\ref{FinslerFermatIntro}), up to an overall constant factor.
}
\end{rmk}

\begin{pro} 
The exterior derivative, $d\varpi_{\mu_0}$, of the $1$-form (\ref{omegaspFinsler}) on $\SO_2(SM)$ descends to the evolution space, $SM$, as
\begin{equation}
\sigma_{\mu_0}
=
p\,h_{ab}\,\omega^\bara\wedge\omega^b
+
\half\,\hOmega(S)
-\frac{1}{2}\,S_{ab}\,\omega^\bara\wedge\omega^\barb,
\label{sigmamu0}
\end{equation}
where the $h_{ab}=\delta_{ab}-\delta^3_a\delta^3_b$ denote the frame-components of the angular metric, and $\hOmega(S)=\hOmega_{ab}\,S^{ab}$ 
the spin-curvature coupling $2$-form.
\end{pro}
\begin{proof}
Let us start with the expression (\ref{omegaspBisFinsler}) of the $1$-form $\varpi_{\mu_0}$.
We have found, see (\ref{domegaH}), that
$d\omega^3=\delta_{AB}\,\omega^\barA\wedge\omega^B$. 
With the help of the structure equations of the Cartan connection, we obtain $d\homega_{12}=\hOmega_{12}+\homega^{\;a}_1\wedge\homega_{a2}=\hOmega_{12}+\homega^{\;3}_1\wedge\homega_{32}=\hOmega_{12}-\homega^{\;1}_3\wedge\homega^{\;2}_3$, using the property~(\ref{omegaCartanSkew}). We then resort to (\ref{omegana}), and to the property (\ref{Al=0}) of the Cartan tensor (that is $A_{ab3}=0$, for all~$a,b=1,2,3$), to find $\homega^{\;A}_3=\omega^\barA$. This yields
$d\homega_{12}=\hOmega_{12}-\omega^{\bar{1}}\wedge\omega^{\bar{2}}$, or, equivalently, $d\homega_{AB}=\hOmega_{AB}-\omega^\barA\wedge\omega^\barB$, for $A,B=1,2$. 
Thus
\begin{equation}
d\varpi_{\mu_0}
=
p\,\delta_{AB}\,\omega^\barA\wedge\omega^B
+
\half\,\hOmega(S)
-\frac{s}{2}\,\epsilon_{AB}\,\omega^\barA\wedge\omega^\barB,
\label{sigmamu0Bis}
\end{equation}
where $\hOmega(S)=s\,\hOmega_{12}=\half{}s\,\epsilon_{AB}\,\hOmega^{AB}$.
In order to prove (\ref{sigmamu0}), we simply use the frame components, $S^{ab}$, of the spin tensor given in (\ref{mu0}). 

To complete the proof, it is enough to verify that $d\varpi_{\mu_0}$, given by (\ref{sigmamu0Bis}), is an integral invariant of the $\SO(2)$-flow generated by 
\begin{equation}
Z=e^i_1\frac{\partial}{\partial{e^i_2}}-e^i_2\frac{\partial}{\partial{e^i_1}}.
\label{Z}
\end{equation}
This is, indeed, the case since $d\varpi_{\mu_0}(Z)=0$.
\end{proof}


\subsubsection{The Finsler-Cartan spin tensor}

Let us now give a construction of the \textit{spin tensor} on the indicatrix-bundle, $SM$, that will be useful in the sequel.

\begin{lem} The following $3$-form of $\SO_2(SM)$, viz.,
\begin{equation}
\Vol=\omega^1\wedge\omega^2\wedge\omega^3
\label{vol}
\end{equation}
is an integral invariant of the flow generated by the vector field, $Z$, given by (\ref{Z}). It descends to $SM$ as $\vol=\frac{1}{6}\vol_{ijk}(x,\dirsupp)\,dx^i\wedge{}dx^j\wedge{}dx^k$, with
\begin{equation}
\vol_{ijk}(x,\dirsupp)=\sqrt{\det{(\rg_{lm}(x,\dirsupp))}}\,\epsilon_{ijk},
\label{volijk}
\end{equation}
where $i,j,k,l,m=1,2,3$.
\end{lem}


\begin{proof}
In view of (\ref{dualOrthonormalBases}), we clearly have $\Vol(Z)=0$. Using, for example, the Chern connection, we find, with the help of the structure equations (\ref{Omega=0}) that $d\Vol=-\omega^{\;a}_a\wedge\Vol$. Then, Equation (\ref{omegaabSym}) readily implies $d\Vol=A^a_{ab}\,\omega^\barb\wedge\Vol$. Again, (\ref{dualOrthonormalBases}) entails that $(d\Vol)(Z)=0$, thus $Z\in\ker(\Vol)\cap\ker(d\Vol)$, proving that $\Vol$ is an $\SO(2)$-integral invariant.

Locally, we have $\vol=\det(\omega^a_i)\,dx^1\wedge{}dx^2\wedge{}dx^3=\sqrt{\det{(\rg_{ij})}}\,dx^1\wedge{}dx^2\wedge{}dx^3$, since~(\ref{gab}) can be rewritten as $\omega^a_i\omega^b_j\,\delta_{ab}=\rg_{ij}$. This proves Equation~(\ref{volijk}).
\end{proof} 

Let us now regard the $S_{ab}$ as the frame-components of a skew-symmetric tensor, $S$, on~$SM$, with components $S_{ij}=S_{ab}\,\omega^a_i\omega^b_j$. Then, owing to (\ref{mu0}), we easily find $S_{ij}=s\,\epsilon_{AB}\,\omega^A_i\omega^B_j=2s\,\omega^1_{[i}\omega^2_{j]}$; those turn out to be nothing but the components of the tensor $S=s\,\vol(\he_3)=s\,\omega^1\wedge\omega^2$. Whence the

\begin{lem}
If we call spin tensor, associated with the moment (\ref{mu0}), the tensor
\begin{equation}
S=s\,\vol(\hdirsupp),
\label{SpinTensor}
\end{equation}
where $\hdirsupp =\he_3$, see (\ref{orthonormalBases}), then
\begin{equation}
S_{ij}=s\,\vol_{ijk}\,\dirsupp^k,
\label{Sij}
\end{equation}
for all $i,j=1,2,3$, where the $\vol_{ijk}$ are as in (\ref{volijk}).
\end{lem}

\subsubsection{Laws of geometrical spinoptics in Finsler-Cartan spaces}
  
We are now ready to determine the explicit expression of the characteristic foliation of the $2$-form (\ref{sigmamu0}) that will provide us with the differential equations governing the trajectories of spinning light in a Finsler-Cartan background.



\begin{lem}\label{hOmegaSLemma}
The spin-curvature coupling term for the Cartan connection retains the form
\begin{equation}
\hOmega(S)=\half\hR(S)_{cd}\,\omega^c\wedge\omega^d+\hP(S)_{cd}\,\omega^c\wedge\omega^\bard+\half\hQ(S)_{cd}\,\omega^\barc\wedge\omega^\bard
\label{hOmegaS}
\end{equation}
where $\hR(S)_{cd}=\hR_{abcd}\,S^{ab}$, etc., and
\begin{eqnarray}
\label{hRS}
\hR(S)_{cd}&=&R(S)_{cd},\\
\label{hPS}
\hP(S)_{cd}&=&P(S)_{cd}=2(A^{}_{cda|b}-A^{}_{ace}\,\dot{A}^e_{bd})S^{ab},\\
\label{hQS}
\hQ(S)_{cd}&=&-2A^{}_{aec}\,A^e_{\;bd}\,S^{ab}.
\end{eqnarray}
\end{lem}
\begin{proof}
The frame-components $\hR(S)_{cd}$, and $\hP(S)_{cd}$, for the Cartan connection exactly match their counterpart for the Chern connection with curvature tensors $R^{\;i}_{j\;kl}$ and $P^{\;i}_{j\;kl}$ given by (\ref{R}) and (\ref{P}) respectively; indeed, the equations~(\ref{hRS}), and~(\ref{hPS}) are derived, in a straightforward way, using~(\ref{hR}), and (\ref{hP}), together with the total skew-symmetry (resp. symmetry) of the spin tensor, (resp. the Cartan tensor), i.e., $S^{ab}=S^{[ab]}$ (resp. $A_{abc}=A_{(abc)}$).

\goodbreak

The constitutive equation (Equation (3.4.11) in \cite{BCS}) for the $\mathrm{hv}$-compo\-nents of the Chern curvature in terms of the covariant derivatives of the Cartan tensor (see~(\ref{nablaellBis}), and (\ref{dotA})), reads
$
P_{abcd}=A_{acd|b}-A_{bad|c}-A_{cbd|a}+A^{}_{bam}\dot{A}^m_{cd}-A^{}_{acm}\dot{A}^m_{bd}+A^{}_{bcm}\dot{A}^m_{ad}.
$
We readily deduce that $P(S)_{cd}=P_{abcd}\,S^{ab}=2(A^{}_{cda|b}-A^{}_{ace}\dot{A}^e_{bd})S^{ab}$, proving the last part of Equation (\ref{hPS}). 

The last equation (\ref{hQS}) is a trivial consequence of (\ref{hQ}).
\end{proof}

\begin{lem}\label{PS3a=0}
There holds $P(S)_{a3}=P(S)_{3a}=0$, for all $a=1,2,3$.
\end{lem}
\begin{proof}
We have $P(S)_{c3}=0$, because of (\ref{Pl=0}). Likewise, the relation $P(S)_{3a}=0$ stems from (\ref{hPS}) since $A_{a3e}=0$ (see~(\ref{Al=0})), and $(e_3)^i_{|j}=0$ (see~(\ref{nablaellBis})).
\end{proof}

We can now proclaim our main result.

\begin{thm}\label{MainTheorem}
The characteristic foliation of the $2$-form $\sigma_{\mu_0}$ of $SM$, given by~(\ref{sigmamu0}), is expressed as follows, viz.,
\begin{equation}
\begin{array}{c}
X\in\ker(\sigma_{\mu_0})\\[10pt]
\Updownarrow\\[10pt]
\begin{array}{ll}
X=X^3&\bigg[
\bigg(\dirsupp^i
+
\displaystyle
\frac{1}{2s\Sigma}\,S^{\;\;i}_jR(S)^{\;\;j}_k\,\dirsupp^k\bigg)\frac{\delta}{\delta{x^i}}\\[14pt]
&
\displaystyle
+\frac{1}{2s^2\Delta\Sigma}\left(S^{\;\;i}_j\left[p\,\delta^{\;j}_k-\half\,P(S)^{\;\;j}_k\right]S^{\;\;k}_l{}R(S)^{\;\;l}_m\,\dirsupp^m\right)\frac{\partial}{\partial{\dirsupp^i}}
\bigg]
\end{array}
\end{array}
\label{kersigmamu0Bis}
\end{equation}
for some $X^3\in\bbR$, where the $S^{\;i}_j=s\,\vol^{\;\;i}_{j\ k}\dirsupp^k$ are as in (\ref{Sij}), and
\begin{eqnarray}
\label{DeltaBis}
\Delta
&=&
s\left[1-\frac{1}{4s^2}\,\hQ(S)(S)\right],\\[6pt]
\nonumber
\label{SigmaBis}
\Sigma
&=&
\frac{1}{\Delta}\left[p^2-\half{}p\,P(S)_{ij}\,\rg^{ij}
+
\frac{1}{8s^2}P(S)_{ik}\,P(S)_{jl}\,S^{ij}S^{kl}\right]\\[6pt]
&&
+\frac{1}{4s}\,R(S)(S),
\end{eqnarray}
with $R(S)(S)=R(S)_{ij}\,S^{ij}$, and $\hQ(S)(S)=\hQ(S)_{ij}\,S^{ij}$.

The $2$-form $\sigma_{\mu_0}$ endows $SM\setminus(\Delta^{-1}(0)\cup\Sigma^{-1}(0))$ with a presymplectic structure of rank $4$; the foliation (\ref{kersigmamu0Bis}) leads to a spin-induced deviation from the geodesic spray (\ref{geodesicSprayBis}), and, according to Axiom \ref{AxiomSpinoptics}, governs spinoptics in a $3$-dimensional Finsler-Cartan structure $(M,F)$.
\end{thm}

\begin{proof}
Using Lemmas \ref{hOmegaSLemma} and \ref{PS3a=0}, we can rewrite our $2$-form (\ref{sigmamu0}) of $SM$, in the guise of (\ref{sigmamu0Bis}), as
\begin{eqnarray}
\nonumber
\sigma_{\mu_0}
&=&
+p\,\delta_{AB}\,\omega^\barA\wedge\omega^B
-\frac{s}{2}\,\epsilon_{AB}\,\omega^\barA\wedge\omega^\barB\\[6pt]
\label{sigmamu0Ter}
&&
+\frac{1}{4}R(S)_{AB}\,\omega^A\wedge\omega^B
+\half{}R(S)_{A3}\,\omega^A\wedge\omega^3\\[6pt]
&&
\nonumber
+\half{}P(S)_{AB}\,\omega^A\wedge\omega^\barB
+\frac{1}{4}\hQ(S)_{AB}\,\omega^\barA\wedge\omega^\barB.
\end{eqnarray}
The proof of Theorem \ref{ThmGeodesicSpray} is adapted to the new $2$-form (\ref{sigmamu0Ter}) we are dealing with. In particular, the vector fields $X\in\Vect(SM)$ 
will be written in the following form, $X=X^A\,\he_A+X^3\,\he_3+X^\barA\,\he_\barA+X^{\bar{3}}\,\he_{\bar{3}}$. Then $X\in\ker(\sigma_{\mu_0})$ iff $\sigma_{\mu_0}(X)+\lambda\omega^{\bar{3}}=0$, where $\lambda\in\bbR$ is a Lagrange multiplier associated with the constraint $F=1$ defining $SM\hookrightarrow\TMM$.
We find 
\begin{eqnarray*}
\sigma_{\mu_0}(X)+\lambda\omega^{\bar{3}}
&=&
+p\,\delta_{AB}(X^\barA\omega^B-X^B\omega^\barA)-s\,\epsilon_{AB}X^\barA\omega^\barB+\lambda\omega^{\bar{3}}\\[6pt]
&&
+\half\,R(S)_{AB}X^A\omega^B+\half\,R(S)_{A3}(X^A\omega^3-X^3\omega^A)\\[6pt]
&&
+\half\,P(S)_{AB}(X^A\omega^\barB-X^\barB\omega^A)+\half\,\hQ(S)_{AB}X^\barA\omega^\barB,
\end{eqnarray*}
so that $X\in\ker(\sigma_{\mu_0})$ iff
\begin{eqnarray}
\label{I}
0&=&\left[p\,\delta^{A}_B-\half\,P(S)^A_{\;\;B}\right]X^\barB+\half\,R(S)^{\;\;A}_BX^B+\half\,R(S)^{\;\;A}_3X^3,\\[6pt]
\label{II}
&=&\left[p\,\delta^{A}_B-\half\,P(S)^{\;\;A}_B\right]X^B
+\left[s\,\epsilon^{\;\;A}_B-\half\,\hQ(S)^{\;\;A}_B\right]X^\barB,\\[6pt]
\label{III}
&=&
R(S)_{A3}X^A,\\[6pt]
\label{IV}
&=&\lambda.
\end{eqnarray}

Put $R(S)(S)=R_{abcd}\,S^{ab}S^{cd}$, and $\hQ(S)(S)=\hQ_{abcd}\,S^{ab}S^{cd}$,
and consider~(\ref{mu0}) to readily get
\begin{equation}
R(S)_{AB}=\frac{1}{2s}\,R(S)(S)\,\epsilon_{AB}
\qquad
\hbox{and}
\qquad
\hQ(S)_{AB}=\frac{1}{2s}\,\hQ(S)(S)\,\epsilon_{AB},
\label{RSABQSAB}
\end{equation}
for all $A,B=1,2$. Note that we also have $R(S)(S)=s^2R_{ABCD}\,\epsilon^{AB}\epsilon^{CD}$, and $\hQ(S)(S)=-2s^2A_{ACE}A_{BDF}\,\epsilon^{AB}\epsilon^{CD}\delta^{EF}$ (see (\ref{hQ})).

Plugging (\ref{RSABQSAB}) into (\ref{II}), we easily find
\begin{equation}
X^\barA=\frac{1}{\Delta}\,\epsilon^{\;\;A}_B\left[p\,\delta^{B}_C-\half\,P(S)^{\;\;B}_C\right]X^C,
\label{XbarA}
\end{equation}
for all $A=1,2$, where
\begin{equation}
\Delta=s\left[1-\frac{1}{4}\,\hQ_{ABCD}\,\epsilon^{AB}\epsilon^{CD}\right].
\label{Delta}
\end{equation}

We then find, with the help of (\ref{XbarA}), and (\ref{I}), the following relationship
\begin{equation}
\Sigma^{\;\;A}_DX^D
=
-\half\,R(S)^{\;\;A}_3X^3,
\label{SigmaX}
\end{equation}
where $\Sigma^{\;\;A}_D=\Delta^{-1}\left[p\,\delta^{A}_B-\half\,P(S)^A_{\;\;B}\right]
\epsilon^{\;\;B}_C\left[p\,\delta^{C}_D-\half\,P(S)^{\;\;C}_D\right]+\half\,R(S)^{\;\;A}_D$. 

\goodbreak

We clearly have $\Sigma_{AB}+\Sigma_{BA}=0$, hence $\Sigma_{AB}=\Sigma\,\epsilon_{AB}$, and find, with some more effort,
\begin{eqnarray}
\nonumber
\Sigma
&=&
\frac{1}{\Delta}\left[p^2-\half{}p\,P(S)_{AB}\,\delta^{AB}
+
\frac{1}{8}P(S)_{AC}\,P(S)_{BD}\,\epsilon^{AB}\epsilon^{CD}\right]\\[6pt]
\label{Sigma}
&&+\frac{1}{4}s\,R_{ABCD}\,\epsilon^{AB}\epsilon^{CD},
\end{eqnarray}
where $\Delta$ is as in (\ref{Delta}).

\goodbreak

Let us point out that Equation (\ref{III}) trivially holds true in view of the skew-symmetry of $\Sigma_{AB}$; indeed, (\ref{SigmaX}) implies $R(S)_{3A}X^A=\Sigma_{AB}X^AX^B=0$.

Equation (\ref{SigmaX}) then leaves us with
\begin{equation}
X^A=\frac{1}{2\Sigma}\,\epsilon^{\;\;A}_B\,R(S)^{\;\;B}_3X^3,
\label{XA}
\end{equation}
for all $A=1,2$. Let us recall that the latter equation for $X^A$ completely determines~$X^\barA$, via (\ref{XbarA}), the components $X^3$, and $X^{\bar{3}}$ remaining otherwise arbitrary. Now, $X\in\Vect(SM)$ if $X(F)=0$, i.e., $\omega^{\bar{3}}(X)=X^{\bar{3}}=0$. We are, hence, left with only one arbitrary parameter, $X^3$, to define the direction~$\ker(\sigma_{\mu_0})$ wherever $\Delta\neq0$, and $\Sigma\neq0$. Thus 
$X=X^A\he_A+X^3\he_A+X^\barA\he_\barA$, where $X^A$, and $X^\barA$ are as in (\ref{XA}), and~(\ref{XbarA}) respectively, with $X^3\in\bbR$.

\goodbreak

Introducing the unit supporting element, $(\dirsupp^a=\delta^{a}_3)$, 
as well as the spin tensor, $(S^{\;\;a}_b=s\,\epsilon^{\;\;A}_B\delta^{a}_A\delta^{B}_b)$, given in (\ref{mu0}), we find
\begin{equation}
\begin{array}{c}
X\in\ker(\sigma_{\mu_0})\\[10pt]
\Updownarrow\\[10pt]
\begin{array}{ll}
X=X^3\bigg[&
\displaystyle
\frac{1}{2s\Sigma}(S\,R(S)\dirsupp)^A\he_A.
+\he_3\\[10pt]
&
\displaystyle
+\frac{1}{2s^2\Delta\Sigma}\left(p\,S-\half\,S\,P(S)\,S\,R(S)\dirsupp\right)^A\he_\barA
\bigg]
\end{array}
\end{array}
\label{kersigmamu0}
\end{equation}
for some $X^3\in\bbR$.

\goodbreak

To complete the calculation, we express (\ref{kersigmamu0}) in terms of the co\-ordinates $\dirsupp^i=e^i_3$ of the distinguished element $u=e_3$, and those, $S^{ij}=\rg^{ik}\rg^{jl}\,S_{kl}$, of the spin tensor $S$, see~(\ref{SpinTensor}), where the $S_{ij}$ are as in (\ref{Sij}). We also bear in mind that $\he_a=e^i_a\,\delta/\delta{x^i}$, and $\he_\bara=e^i_a\,\partial/\partial{\dirsupp^i}$, for all $a=1,2,3$, as given by (\ref{orthonormalBases}), on the principal bundle $\SO_2(SM)$ above the evolution space $SM$. The upshot of the computation is that the characteristic foliation (\ref{kersigmamu0}) can be recast in the form~(\ref{kersigmamu0Bis}); Equations~(\ref{DeltaBis}), and (\ref{SigmaBis}) also provide alternative expressions for~(\ref{Delta}), and~(\ref{Sigma}). 

At those points $(x,u)\in{}SM$ where $\Delta=0$, or $\Sigma=0$, singularities of the foliation~(\ref{kersigmamu0}) do occur; they must be discarded to guarantee a well-behaved pre\-symplectic structure of (generic) rank $4$. The proof is now complete. 
\end{proof}

\begin{rmk}
{\rm
Let us choose, e.g., $X^3=1$ in (\ref{kersigmamu0Bis}) to define the generator, $X$, of the foliation $\ker(\sigma_{\mu_0})$. The latter significantly deviates from a spray since the velocity, $\dx$, given by the horizontal projection of $X$, differs from the direction, $\dirsupp$, of the supporting element, namely
\begin{equation}
\dx^i=\dirsupp^i
+
\frac{1}{2s\Sigma}\,S^{\;\;i}_jR(S)^{\;\;j}_k\,\dirsupp^k,
\label{anomalousVelocity}
\end{equation}
where $\dx^i=X(x^i)$, for all $i=1,2,3$. The occurrence of this \textit{anomalous velocity} in the presence of curvature can be classically interpreted (see \cite{DHH,DHH2}) as the source of the \textit{optical Hall effect}.
Moreover, the vertical components of the generator $X$, namely those of the geodesic acceleration, depend linearly on the helicity, $\varepsilon=\sign(s)$. They, notably, lead to a splitting, \textit{\`a la} Stern-Gerlach, of light rays with opposite helicities.
}
\end{rmk}

Let us finish with the following corollary of Theorem \ref{MainTheorem} which help us recover the simpler equations of spinoptics in the Riemannian case, derived in \cite{DHH}.

\begin{cor}\label{RiemannCor}
If the Finsler structure, $(M,F)$, is Riemannian, the characteristic foliation of the $2$-form $\sigma_{\mu_0}$ is spanned  by the vector field
\begin{equation}
X=
\bigg(\dirsupp^i
+
\displaystyle
\frac{1}{2\Sigma'}\,S^{\;\;i}_jR(S)^{\;\;j}_k\,\dirsupp^k\bigg)\frac{\delta}{\delta{x^i}}
-\frac{1}{2\Sigma'} R(S)^{\;\;i}_j\,\dirsupp^j\frac{\partial}{\partial{\dirsupp^i}}
\label{kersigmamu0Ter}
\end{equation}
with
\begin{equation}
\Sigma'
=p^2+\frac{1}{4}\,R(S)(S),
\label{Sigma'}
\end{equation}
where the $R_{ijkl}$ are the components of the Riemann curvature tensor.
\end{cor}

\begin{proof}
Suffice it to note that the Cartan tensor vanishes iff the Finsler structure is Riemannian, hence $P_{ijkl}=\hQ_{ijkl}=0$. The curvature tensor $R_{ijkl}$ in (\ref{kersigmamu0Bis}) then reduces to the Riemann curvature tensor, see (\ref{R}). The proof is completed by noticing that $S_{jk}S^{ki}=s^2(\dirsupp^i\dirsupp_j-\delta^{i}_j)$, a direct consequence of~(\ref{Sij}).
\end{proof}


\section{Conclusion and outlook}\label{SectionConclusion}

We have proposed a generalization of the Fermat Principle enabling us to describe spinning light rays in a general, non dispersive, optical medium, namely an in\-homogeneous and anisotropic medium modeled on a Finsler manifold. The guideline for this extension has been provided by the Euclidean symmetry of the free system, viewed as a generic coadjoint orbit of the Euclidean group, $\SE(3)$. Inter\-action with the optical medium has been justified in terms of a minimal coupling of the model to the (affine) Cartan connection of the Finsler structure; the gist of the procedure lies in the fact that the affine Cartan connection takes, indeed, its values in the Lie algebra $\se(3)$ and, thus, couples naturally to the moment $\mu_0\in\se(3)^*$ defining the original coadjoint orbit (the classical states of the free Euclidean photon). The resulting presymplectic structure on (an open submanifold of) the indicatrix-bundle has been investigated. In particular the characteristic foliation of this structure has been worked out, and shown to yield a system of differential equations governing the trajectories of spinning light rays, associated with a vector field departing from the usual Finslerian geodesic spray. The geodesic acceleration of spinning light rays is due to the coupling of spin with the Finsler-Cartan curvature, which also engenders an anomalous velocity. The latter, already present in Riemannian spinoptics \cite{DHH,DHH2} has proved crucial in the geometrical interpretation of the brand new optical Hall effect, see, e.g., \cite{BB04bis,OMN}. The consubstantial nature of this effect with the geometry of Euclidean coadjoint orbits is precisely what prompted the present study, and our endeavor to depart from the case of isotropic media by taking advantage of Finsler-Cartan structures. Although the characteristic foliation (\ref{kersigmamu0Bis}) of the above-mentioned presymplectic structure is of a formidable complexity, it is nevertheless a mandatory consequence of a minimal, geometrically justified, modification (\ref{FinslerFermatIntro}) of the Hilbert $1$-form (\ref{FinslerFermat2form}) of central importance in Finsler geometry.

\goodbreak

The future perspectives opened by this work are manifold. 

It would be desirable to linearize the differential equations of Finsler spin\-optics in the case of weakly curved Finsler-Cartan manifolds, to account for weakly anisotropic optical media. This should lead to substantial simplifications, suitable for an explicit calculation of the geodesic deviation in several non trivial examples, such as those given by (\ref{Fo}), and (\ref{Fe}). Also, would it be of great importance to compare this linearized set of differential equations with the outcome of the calculations performed by a (short wavelength) semi-classical limit of the Maxwell equations in weakly anisotropic and inhomogeneous media~\cite{BFK}.

The Fermat Principle has, most interestingly, been generalized, via a novel variational calculus, to the case of lightlike geodesics in Finsler spacetimes with a Lorentzian signature \cite{Per}. It would be worth investigating how that relativistic version of geometrical optics extends to spinoptics in relativistic Finsler spacetimes. Let us note that Randers Finsler metrics 
play, as discussed in \cite{CJM}, a prominent r\^ole in such a framework, corresponding to induced (instantaneous) Finsler metrics on the material body of the optical medium.

At last, specific applications of the equations of Finsler-Cartan spinoptics should be explored in a number of other directions such as the Kerr, the Faraday effects, and the Cotton-Mouton effect responsible for plasma birefringence, see, e.g., \cite{SZ}, as well as the photonic Hall effect \cite{Tig} in the presence of a magnetic field.

In truth, the present study of Finsler spinoptics was a challenge, taken up from a purely geometric standpoint; one may, conceivably, expect it will provide further insights into modern trends of geo\-metrical optics of anisotropic media.

\newpage



\end{document}